%% file: main.tex
\documentclass[letterpaper,11pt]{article}
\usepackage[utf8]{inputenc}

\usepackage{amsmath, amsthm, amssymb, thm-restate}
\usepackage{algorithmicx}
\usepackage[table,xcdraw]{xcolor}
\usepackage[vlined,linesnumbered]{algorithm2e}
\usepackage{setspace}
\usepackage{mathtools}
\usepackage[numbers]{natbib}
\usepackage{comment} 
\usepackage{tcolorbox} 
\usepackage{xfrac}
\usepackage{hyperref}
\usepackage{multirow}
\usepackage{caption}
\usepackage{bm}
\usepackage{newfloat}
\usepackage{enumitem}
\usepackage{dblfloatfix} 
\usepackage{wrapfig}

\usepackage{fancyhdr}

\usepackage[margin=1in]{geometry}

\allowdisplaybreaks

\definecolor{mygreen}{RGB}{10,110,230}
\definecolor{myred}{RGB}{10,110,230}

\hypersetup{
     colorlinks=true,
     citecolor= mygreen,
     linkcolor= myred
}

\input{macro.tex}
\input{utils.tex}

\makeatletter
\renewcommand{\paragraph}{%
  \@startsection{paragraph}{4}%
  {\z@}{10pt}{-1em}%
  {\normalfont\normalsize\bfseries}%
}
\makeatother

\makeatletter
\patchcmd{\@algocf@start}
  {-1.5em}
  {0pt}
  {}{}
\makeatother

\title{Beating Greedy Matching in Sublinear Time}

\author{
Soheil Behnezhad\footnote{Supported by NSF CAREER Award 1942123, NSF Award 1812919, NSF CCF-1954927, a Simons Investigator Award, and a David and Lucile Packard Fellowship.}\\
\and 
Mohammad Roghani\\
\and
Aviad Rubinstein\footnote{Supported by NSF CCF-1954927, and a David and Lucile Packard Fellowship.}
\and
Amin Saberi
}

\date{}

\begin{document}

\maketitle

\thispagestyle{empty}
\input{abstract}

\clearpage

{
\hypersetup{hidelinks}
\vspace{1cm}
\renewcommand{\baselinestretch}{0.1}
\setcounter{tocdepth}{2}
\tableofcontents{}
\thispagestyle{empty}
\clearpage
}

\setcounter{page}{1}
\input{intro.tex}

\input{prelim.tex}

\input{algorithm.tex}
\input{query_process.tex}

\input{conclusion.tex}

\bibliographystyle{alpha}
\bibliography{references}

\clearpage
\appendix

\input{appendix-proofs}
	
\end{document}

%% file: macro.tex
\renewcommand{\epsilon}{\varepsilon}

\DeclareMathOperator{\E}{\ensuremath{\normalfont \textbf{E}}}

\newcommand{\hiddencomment}[1]{}

\newcommand{\cev}[1]{\reflectbox{\ensuremath{\vec{\reflectbox{\ensuremath{#1}}}}}}

\newcommand{\greedyMM}[1]{\ensuremath{\mathsf{GreedyMM}(#1)}}
\newcommand{\start}[0]{\ensuremath{\textsc{Start}}}
\newcommand{\extend}[0]{\ensuremath{\textsc{Extend}}}

\newcommand{\rankstart}[1]{\ensuremath{\pi_{\start}(#1)}}
\newcommand{\rankextend}[1]{\ensuremath{\pi_{\extend}(#1)}}

\newcommand{\true}[0]{\ensuremath{\textsc{True}}}
\newcommand{\false}[0]{\ensuremath{\textsc{False}}}
\newcommand{\VO}[0]{\ensuremath{\textsc{VO}}}

\newcommand{\EOS}[0]{\ensuremath{\textsc{EOS}}}
\newcommand{\EOE}[0]{\ensuremath{\textsc{EOE}}}

\newcommand{\MS}[0]{\ensuremath{\textup{MS}}}
\newcommand{\lowest}[0]{\ensuremath{\textsc{Lowest}}}

\newcommand{\blue}[0]{\ensuremath{\textsc{Blue}}}
\newcommand{\red}[0]{\ensuremath{\textsc{Red}}}

\SetKw{Continue}{continue}

%% file: utils.tex
\DeclareMathOperator{\poly}{poly}

\usepackage[noabbrev,nameinlink]{cleveref}
\crefname{lemma}{Lemma}{Lemmas}
\crefname{theorem}{Theorem}{Theorems}
\crefname{property}{Property}{Properties}
\crefname{claim}{Claim}{Claims}
\crefname{definition}{Definition}{Definitions}
\crefname{observation}{Observation}{Observations}
\crefname{proposition}{Proposition}{Propositions}
\crefname{assumption}{Assumption}{Assumptions}
\crefname{line}{Line}{Lines}
\crefname{figure}{Figure}{Figures}
\creflabelformat{property}{(#1)#2#3}
\crefname{equation}{}{}
\crefname{section}{Section}{Sections}
\crefname{appendix}{Appendix}{Appendices}
\crefname{algCounter}{Algorithm}{Algorithms}
\Crefname{algCounter}{Algorithm}{Algorithms}

\newtheorem{theorem}{Theorem}[section]
\newtheorem{question}{Question}
\newtheorem{lemma}[theorem]{Lemma}
\newtheorem{proposition}[theorem]{Proposition}
\newtheorem{corollary}[theorem]{Corollary}

\newtheorem{definition}[theorem]{Definition}
\newtheorem{claim}[theorem]{Claim}

\newtheorem{observation}[theorem]{Observation}
\newtheorem{remark}[theorem]{Remark}
\newtheorem{assumption}[theorem]{Assumption}

\definecolor{mylightgray}{RGB}{230,230,230}

\algnewcommand{\IIf}[2]{\textbf{if} #1 \textbf{then} #2}
\algnewcommand{\EndIIf}{\unskip\ \algorithmicend\ \algorithmicif}

\newenvironment{graytbox}{
\par\addvspace{0.1cm}
\begin{tcolorbox}[width=\textwidth,
                  boxsep=5pt,
                  left=1pt,
                  right=1pt,
                  top=2pt,
                  bottom=2pt,
                  boxrule=0pt,
                  arc=0pt,
                  colback=mylightgray,
                  colframe=black,
                  ]
}{
\end{tcolorbox}
}

\newenvironment{whitetbox}{
\par\addvspace{0.1cm}
\begin{tcolorbox}[width=\textwidth,
                  boxsep=5pt,
                  left=1pt,
                  right=1pt,
                  top=2pt,
                  bottom=2pt,
                  boxrule=1pt,
                  arc=0pt,
                  colframe=black,
                  colback=white
                  ]
}{
\end{tcolorbox}
}

\newenvironment{myproof}{
\vspace{-0.5cm}
\begin{proof}
}{
\end{proof}
}

\newcounter{algCounter}

\newenvironment{algenv}[2]{
    \begin{whitetbox}
        \refstepcounter{algCounter}
        \textbf{Algorithm~\thealgCounter:} #1
        \label{#2}
        
        \vspace{-0.2cm}
        \noindent\rule{\linewidth}{1pt}
        \begin{algorithm}[H]
}{
        \end{algorithm}
    \end{whitetbox}
}

%% file: abstract.tex
\begin{abstract}
    We study sublinear time algorithms for estimating the size of maximum matching in  graphs. Our main result is a $(\frac{1}{2}+\Omega(1))$-approximation algorithm which can be implemented in $O(n^{1+\epsilon})$ time, where $n$ is the number of vertices and the constant $\epsilon > 0$ can be made arbitrarily small. The best known lower bound for the problem is $\Omega(n)$, which holds for any constant approximation.
    
    \smallskip\smallskip
    Existing algorithms either obtain the greedy bound of $\frac{1}{2}$-approximation [Behnezhad FOCS'21], or require some assumption on the maximum degree to run in $o(n^2)$-time [Yoshida, Yamamoto, and Ito STOC'09]. We improve over these by designing a less ``adaptive'' augmentation algorithm for maximum matching that might be of independent interest.
\end{abstract}

%% file: intro.tex
\section{Introduction}

Linear-time algorithms have long been considered the gold standard in algorithm design. With the rapid increase in the size of data, however, even linear-time algorithms may be slow in some settings. A natural question is whether it is possible to solve a problem of interest in {\em sublinear time} in the input size. That is, without even reading the whole input. In this work, we focus on the problem of estimating the size of {\em maximum matching} in sublinear time. Recall that a {\em matching} is a set of edges no two of which share an endpoint, and a {\em maximum matching} is a matching of the largest size. This is a central problem in the study of sublinear time algorithms and several general techniques of the area have emerged from the study of matchings \cite{ParnasRon07,NguyenOnakFOCS08,YoshidaYISTOC09,OnakSODA12,KapralovSODA20,ChenICALP20,behnezhad2021}.

There is a simple greedy algorithm for constructing a {\em maximal} (but not necessarily {\em maximum}) matching: Iterate over the edges and greedily add  each edge whose endpoints are yet unmatched. Every maximal matching is at least half the size of a maximum matching. Thus, this simple algorithm is a $1/2$-approximation. Note, however, that the algorithm's running time is not sublinear in the input size as we have to go over all the possibly $\Omega(n^2)$ edges one by one where $n$ is the number of vertices. In fact, there are lower bounds showing that $\Omega(n^2)$ time is necessary to {\em find} any $O(1)$-approximate matching. Nonetheless, approximating the {\em size} of the maximum matching can be done much faster. 
Indeed, a beautiful line of work in the literature \cite{NguyenOnakFOCS08,YoshidaYISTOC09,OnakSODA12,behnezhad2021} led to an $\widetilde{O}(n)$ time algorithm for estimating the size of a (random) greedy maximal matching \cite{behnezhad2021}.


Unfortunately, the main shortcoming of the greedy maximal matching algorithm is that it only provides a $1/2$-approximation, even when the edges are processed in a random order \cite{DyerF91}. There are techniques to improve the approximation by ``augmenting'' the greedy matching \cite{YoshidaYISTOC09}, but such techniques only work well when the degrees in the graph are rather small. In particular, when we go even slightly above $1/2$-approximation, then all known algorithms take a (large) polynomial time in the maximum degree. Unfortunately, this can be as large as $\Omega(n^2)$ for general $n$-vertex graphs which is no longer sublinear time in the input size. This raises a natural question:
\begin{question}\label{question}
    Is it possible to $(\frac{1}{2}+\Omega(1))$-approximate maximum matching size in $n^{2-\Omega(1)}$ time?
\end{question}

\noindent We remark that this question has been open even for bipartite graphs.


In this work, we answer \Cref{question} in the affirmative. The running time of our algorithm can, in fact, be made abritrarily close to linear in $n$. Our algorithm can be adapted to both the {\em adjacency list} and {\em adjacency matrix} query models, which are the two standard graph representations studied in the literature. We also consider both multiplicative approximations as well as multiplicative-additive approximations. See \cref{sec:prelim} for the formal definitions of these query models and approximations.

Using $n$, $m$, $\Delta$, and $\bar{d}$ to respectively denote the number of vertices, the number of edges, the maximum degree, and the average degree in the graph, our results can be summarized as follows:

\begin{graytbox}
\begin{theorem}\label{thm: main-theorem}
    For any constant $\epsilon > 0$, there is a constant $\delta > 2^{-O(1/\epsilon)}$ along with an algorithm that w.h.p. estimates the size of maximum matching up to a:
    \begin{enumerate}[leftmargin=20pt,label=$(\arabic*)$,itemsep=0pt,topsep=5pt]
        \item multiplicative factor of $(\frac{1}{2} + \delta)$ in the \underline{adjacency list} model in $\widetilde{O}(n+\Delta^{1+\epsilon})$ time,\label{mainthm:adjlist-mult}
        \item multiplicative-additive factor of $(\frac{1}{2} + \delta, o(n))$ in the \underline{adjacency list} model in $\widetilde{O}(\bar{d} \cdot \Delta^\epsilon)$ time,\label{mainthm:adjlist-mult-add}
        \item multiplicative-additive factor of $(\frac{1}{2} + \delta, o(n))$ in the \underline{adjacency matrix} model in $O(n^{1+\epsilon})$ time.\label{mainthm:adjmatrix}
    \end{enumerate}
\end{theorem}
\end{graytbox}

\smallskip

\noindent A few remarks about the three results of \cref{thm: main-theorem}:
\begin{itemize}[itemsep=2pt,topsep=3pt]
    \item The constant $\delta > 0$ in \cref{thm: main-theorem} is miniscule. We did not attempt to optimize it, but do not expect our techniques to lead to a better than, say, .51-approximation in $n^{2-\Omega(1)}$ time.
    \item \cref{thm: main-theorem}--\ref{mainthm:adjlist-mult} comes close to an $\Omega(n)$ lower bound that holds for any $O(1)$-approximation in the adjacency list model. Any such algorithm must distinguish an empty graph from one with a single edge. This clearly requires $\Omega(n)$ queries in the adjacency list model.
    \item \cref{thm: main-theorem}--\ref{mainthm:adjlist-mult-add} comes close to an $\Omega(\bar{d})$ lower bound for any $(O(1), o(n))$-approximation in the adjacency list model due to \cite{ParnasRon07}. Observe that $\bar{d} \Delta^\epsilon \ll m$ for any $\epsilon < 1$. Therefore, this algorithm {\em always} runs in sublinear time in the number of edges in the graph.
    \item \cref{thm: main-theorem}--\ref{mainthm:adjmatrix} comes close to an $\Omega(n)$ lower bound for any $(O(1), o(n))$-approximation in the adjacency matrix model. Any such algorithm must distinguish an empty graph from one that includes a random perfect matching. This requires $\Omega(n)$ adjacency matrix queries.
    \item It takes $\Omega(n^2)$ queries to the adjacency matrix to distinguish an empty graph from one with only a single edge. Since this must be done for any multiplicative $O(1)$-approximation, no non-trivial such algorithm (i.e., one with $o(n^2)$ queries) exists for this model.
\end{itemize}

\vspace{-0.3cm}
\subsubsection*{On Beating Greedy Matching in Various Settings} 
\vspace{-0.2cm}

The greedy 1/2-approximation is a prevalent barrier for maximum matching across various settings. As a result, numerous works in the literature study the possibility of beating it --- both on the upper bound side as well as the lower bound side. 
The answer is not always the same. For instance, for the online model under {\em edge arrivals}, \cite{GamlathKMSW19} showed that 1/2 is provably the best achievable approximation, which can be trivially matched by the greedy algorithm. 
There are also settings where the answer remains unknown, despite a significant research effort. For instance, in the single-pass streaming setting, beating the greedy 1/2-approximation in $\widetilde{O}(n)$ (or even subquadratic) space has been open for nearly two decades \cite{FeigenbaumKMSZ05}, and is often considered as one of the most fundamental open problems of the area.  Finally, there are settings for which the greedy 1/2-approximation has been broken. Various models of the online setting \cite{Fahrbach0TZ20,GamlathKMSW19}, the random-order streaming setting \cite{KonradMM12}, and the stochastic matching setting \cite{AssadiKL17} are examples of this. In many of these settings, the approximation has been  improved well beyond 1/2 after the greedy bound was first broken. For instance, in the random-order streaming the current best known bound is slightly above 2/3 \cite{AssadiB21}, and in the stochastic matching setting a $(1-\epsilon)$-approximation has been achieved \cite{BehnezhadDH20}. We hope that our work in this paper also inspires future work on going tangibly above 1/2-approximation in the sublinear time model.

Discovering short {\em augmenting paths} has been a central technique in many of the works discussed above in beating the greedy algorithm. What varies significantly is whether it is possible to find these augmenting paths effeciently in the particular model at hand. In particular, a common approach is to first construct a maximal matching in full, and then augment it via the vertices left unmatched. This ``adaptivity'' complicates things in our model, making it hard to estimate the size of the solution in subquadratic time. One of our main contributions in this work is to give a less ``adaptive'' algorithm that interleaves the construction of a maximal matching and the augmentation phase. We believe this technique, which is overviewed in \cref{sec:highlevel}, might be of independent interest. 




\section{Technical Overview}\label{sec:highlevel}
%

 In \cref{sec:overview-prior}, we first give a brief overview of existing approaches and discuss why the take quadratic time to implement. Then, in \cref{sec:overview-ours}, we overview our main tool in breaking this quadratic time barrier through a less ``adaptive'' augmentation algorithm.

\subsection{The Quadratic Barrier: A Brief Discussion of Earlier Techniques}\label{sec:overview-prior}

 Discovering (short) {\em augmenting paths}\footnote{See \cref{sec:prelim} for the formal definition of augmenting paths.} is a natural way of improving the approximation for the maximum matching problem. For example, the Hopcroft-Karp \cite{HopcroftK73} algorithm starts with an empty matching and iteratively applies a maximal set of vertex disjoint (short) augmenting paths. Each step of Hopcroft-Karp can essentially be viewed as a {\em maximal independent set} (MIS) instance. Put one vertex for each (short) augmenting path and connect two vertices if their corresponding augmenting paths share a vertex. An MIS in this graph corresponds to a maximal set of vertex disjoint augmenting paths. Building on this idea and by giving a size estimator for MIS, Yoshida, Yamamoto, and Ito~\cite{YoshidaYISTOC09} showed that for any integer $k \geq 1$, a $(\frac{k}{k+1}, o(n))$-approximation can be obtained in $O_{k}(\Delta^{6k(k+1)})$ time\footnote{$O_k$ ignores dependencies on $k$.}. This is sublinear when $\Delta$ is sufficiently small.

Although the MIS-based approach is powerful enough to go well beyond 1/2-approximation in $\poly(\Delta)$ time, it does not seem to help with general graphs where $\Delta$ can be large. This holds even if we limit ourselves to length-3 augmenting paths, which is needed for beating 1/2. The MIS size estimator of \cite{YoshidaYISTOC09} crucially requires time at least linear in the average degree. Since every edge of a maximal matching can belong to $\Omega(\Delta^2)$ length-3 augmenting paths (with $\Omega(\Delta)$ choices from each endpoint of the edge), this average degree in the MIS graph can be $\Omega(\Delta^2)$ where $\Delta$ is the original graph's maximum degree. This makes it unlikely for this approach to yield an $o(\Delta^2)$ time algorithm. Additionally, not being able to construct the MIS graph explicitly in whole, and not having the edges of the maximal matching we are trying to augment also impose other $\Delta$ factors in the running time, arriving at the rather large $\poly(\Delta)$ bound of \cite{YoshidaYISTOC09}.

There is an alternative way of discovering length-3 augmenting path which works directly with matchings instead of independent sets. The idea is to first find a maximal matching $M$ of $G$, then find another maximal matching $S$ on a subgraph $H$ of $G$ which includes a subset of the edges that have exactly one endpoint matched by $M$. Note that if both endpoints of an edge $e \in M$ are matched in $S$, then we get a length-3 augmenting path. This framework was first used by \cite{KonradMM12} in the context of random-order streaming algorithms, but has since been applied to various other settings \cite{BhattacharyaHN16,BehnezhadLM20,GuruganeshS17}. Because the second graph $H$ is defined {\em adaptively}  based on  $M$, we cannot simply run two independent instances of existing maximal matching estimators as black-box. In fact, this framework also hits a quadratic-in-degree time barrier as we describe next. To describe this barrier, we first briefly overview the key ideas behind the maximal matching size estimator of \cite{behnezhad2021}.

Consider a maximal matching $S$ that is constructed greedily by iterating over the edges in some ordering $\pi$. It is not hard to see that $e \in S$ iff there is no edge $e'$ incident to $e$ such that $\pi(e') < \pi(e)$ and $e' \in S$. Therefore to determine whether $e \in S$, it suffices to go over the lower rank neighboring edges of $e$ (in the increasing order of their ranks) and recursively query them to either find one that belongs to $S$, or conclude that $e$ must be in $S$. This approach was first suggested by Nguyen and Onak \cite{NguyenOnakFOCS08}. The main question is the total number of recursive calls needed for the process to finish. The result of Behnezhad \cite{behnezhad2021} is that for a {\em random} vertex $v$ and a {\em random} permutation $\pi$, the process terminates in $O(\bar{d} \cdot \log n)$ expected total time, coming close to an $\Omega(\bar{d})$ lower bound. 

Now, let us revisit the above two-step algorithm  which first constructs a maximal matching $M$ of $G$ and then another maximal matching $S$ of a subgraph $H$ of $G$. Consider the task of determining whether a vertex $v$ is matched by $S$. From \cite{behnezhad2021}, we get that for a random vertex $v$, this should be doable by exploring $\widetilde{O}(\bar{d}_H)$ edges of $H$ in expectation where here we use $\bar{d}_X$ to denote the average degree of graph $X$. The challenge, however, is that since $H$ is defined adaptively based on $M$, we do not a priori know the neighbors of $v$ in $H$. In particular, to know whether an edge $(u, v)$ exists in $H$, we have to ensure that exactly one of $u$ and $v$ is matched in $M$. Therefore, for exploring $\widetilde{O}(\bar{d}_H)$ edges in $H$, the naive approach would make $\widetilde{O}(\bar{d}_H)$ vertex queries to $M$. Note that these calls are not necessarily to random vertices anymore, which is crucial for the bound of \cite{behnezhad2021} to work. But even if we manage to get an $\widetilde{O}(\bar{d}_G)$ time bound on each one of these calls we arrive at a total running time of $\widetilde{\Theta}(\bar{d}_H \cdot \bar{d}_G)$ which again can be as large as $\Omega(n^2)$.

\subsection{Our Contribution: A Less ``Adaptive'' Augmentation Algorithm} \label{sec:overview-ours}

The two-stage algorithm we discussed earlier, constructs $M$ fully and then adaptively picks the edges of $S$ based on $M$. Our first step towards proving \cref{thm: main-theorem} is introducing a less ``adaptive'' algorithm that interleaves the construction of the two matchings $M$ and $S$.

Our algorithm starts by constructing a sequence $T$, containing a single element $(e, \extend)$  and $K \geq 1$ distinct elements $(e, \start)$,  corresponding to every edge $e$ in the graph,  where $K$ is a parameter of the algorithm. We will process $T$ in a {\em random order}. The role of having multiple copies of the \start{} elements is to bias them to appear earlier in the random permutation.  We start by initializing two empty matchings $M$ and $S$ and then iterate over these randomly sorted $m(K+1)$ elements of $T$. Whenever we see an $(e, \start)$ element, we add $e$ to $M$ iff both endpoints of $e$ are unmatched in $M$. Therefore,  $M$ will be a random greedy maximal matching of $G$. Whenever we see an $(e, \extend)$ element, we  add $e$ to $S$ under a few conditions. The first condition is that both endpoints of $e$ must be yet unmatched by $S$; this is to ensure that $S$ continues to be a matching. The second condition is that at most one endpoint of $e$ can be already matched by $M$. Our final algorithm, formalized as \cref{alg: algorithm}, also checks two more technical conditions before adding $e$ to $S$ which are needed for the approximation ratio analysis. Once all of $T$ is processed, the algorithm returns a maximum matching of $M \cup S$.

Observe that even though at the time of adding an edge $(u, v)$ to $S$, at most one of its endpoints is matched in $M$, the other endpoint may get matched in $M$ later in the process. Such edges of $S$ cannot be used in length-3 augmenting paths for $M$. One way to avoid these bad events is to set $K$ large enough so that  all the \start{} elements  appear before all \extend{} elements. However, this will negatively impact the running time. Specifically, the local query process to determine whether a random vertex $v$ is matched in either of $S$ or $M$ takes $\widetilde{O}(\bar{d} \cdot K)$ time. This means that we need to set $K$ to be much smaller than $\Delta$ to beat the quadratic-time-barrier. Indeed we set $K \approx \Delta^\epsilon$ to get the bounds of \cref{thm: main-theorem}.

\begin{wrapfigure}[7]{r}{0.27\textwidth}
    \centering
    \vspace{-0.5cm}
    \includegraphics[scale=0.9]{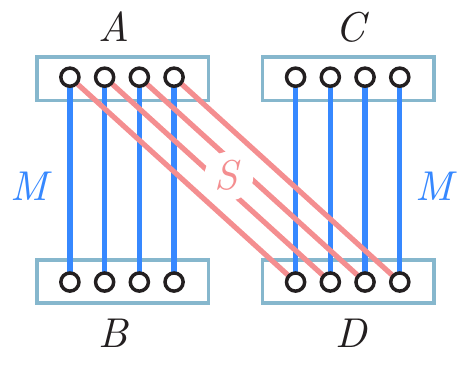}
\end{wrapfigure}
When $K \ll \Delta$, we will inevitably have many edges of $S$ for which both endpoints are matched in $M$. For example, consider the construction illustrated on the right with four vertex parts $A$, $B$, $C$, $D$.
There is a regular bipartite graph of degree $\Theta(\Delta)$ between $A$ and $B$, a regular bipartite graph of degree $\Theta(\Delta^{0.99})$ between $A$ and $D$, and a regular bipartite graph of degree $\Theta(\Delta^{0.98}/K)$ between $C$ and $D$. In this construction, the edges of $M$ match $A$ to $B$ and $C$ to $D$ nearly completely. 
The edges of $S$ match $A$ to $D$ nearly completely. 
This happens because there are many more \extend{} elements in $(A, D)$ than there are $\start{}$ elements in $(C, D)$. Therefore, at the time of adding $S$, the part of $M$ from $C$ to $D$ is not yet constructed. 

Despite this bad event, we show that the edges of $S$ are still useful in augmenting $M$. One key insight apparent in the above example  is that the edges of $(C, D)$ in $M$ tend to have rank (in the permutation of $T$) roughly $K$ times larger than those edges of $(A, B)$ in $M$. We generalize this to all graphs and show that if an edge of $S$ connects two edges of $M$, then the ranks of these edges of $M$ must differ by a factor of roughly $K$.\footnote{To be more precise, we only prove this for most edges of $S$, but not all.} Therefore, if instead of considering all edges of $M$ for augmentation, we consider a subset $M_{j^\star}$ of $M$ that is a constant fraction of $M$ and at the same time any two edges in $M_{j^\star}$ have ranks within $K$ factor of each other, then $S$ cannot connect two edges of $M_{j^\star}$ and so we can focus on augmenting just $M_{j^\star}$ instead of the whole matching $M$. See \cref{sec:algorithm,sec:approx} for the full details of the algorithm and its approximation analysis.

So far we have only described an algorithm that finds a $(\frac{1}{2}+\Omega(1))$-approximate matching and have not yet described how to estimate its output size in sublinear time. To do this, we first define a query process akin to the one described above for maximal matching. That is, for a given edge $e$ we define two query processes that respectively return whether $e \in M$ and $e \in S$. We then analyze the expected number of the recursive calls for a random start vertex by building on the techniques of \cite{behnezhad2021,YoshidaYISTOC09}. Several challenges arise along the way that are unique to our algorithm and require new ideas. For instance, an \extend{} element remains relevant (i.e., can still be added to $S$) until seeing two \start{} elements in $M$, one from each endpoint. This is unlike the greedy approach, say for MIS or maximal matching, where an element (respectively a vertex and an edge) becomes irrelevant right after seeing one neighbor in the solution. Also note that we should not simply count the number of edges in $M$ and $S$. Rather, we have to count the number of edges of $M$ plus the number of length-3 augmenting paths that we find. To do this, when we find an edge $(u, v) \in S$ and an edge $(v, w) \in M$, we also query whether $w$ is matched in $S$ to a vertex left unmatched by $M$ or not. This complicates the analysis because this vertex $w$ is not picked uniformly at random anymore. The details of the query process and its analysis are provided in \cref{sec: query-process}.

%% file: prelim.tex
\section{Preliminaries}\label{sec:prelim}

\paragraph{Notation:} Throughout the paper we use $G=(V, E)$ to denote the input graph. We use $n$ to denote the number of vertices in $G$, $m$ to denote the number of edges in $G$, $\Delta$ to denote the maximum degree of $G$, and $\bar{d}$ to denote the average degree of $G$. We write $\mu(G)$ to denote the size of the maximum matching in $G$. 

We use $A \oplus B := (A \cup B) \setminus (A \cap B)$ to denote the symmetric difference of two sets $A$ and $B$. Also for any positive integer $k$, we use $[k]$ to denote the set $\{1, \ldots, k\}$. Throughout the paper, we use the $\widetilde{O}(\cdot)$ to suppress $\poly\log n$ factors, that is $\widetilde{O}(f) = O(f \cdot \poly\log(n))$.

\paragraph{Problem Definition:} Given a graph $G$, represented in one of the following two ways, we study the problem of estimating the size of maximum matching:
\begin{itemize}
    \item {\em Adjacency List:} In this model, the neighbors of each vertex are stored in a list sorted in an arbitrary order. Each query of the algorithm specifies a vertex $v$ and an index $i$. The answer is the ID of the $i$-th vertex in the list of $v$'s neighbors, or {\em empty} if $v$ has less than $i$ neighbors.
    \item {\em Adjacency Matrix:} In this model, each query of the algorithm specifies a pair of vertices $u$ and $v$. The answer is 1 if $u$ and $v$ are adjacent, and 0 otherwise.
\end{itemize}
For $\alpha \in (0, 1]$ and $\gamma \in [0, 1]$, we say $\widetilde{\mu}(G)$ is a multiplicative-additive $(\alpha, \gamma n)$-approximation of the size of maxmimum matching of $G$ if 
$
\alpha\mu(G) - \gamma n \leq \widetilde{\mu}(G) \leq \mu(G)
$. Additionally, it is a multiplicative $\alpha$-approximation if $\alpha\mu(G) \leq \widetilde{\mu}(G) \leq \mu(G)$.

\paragraph{Augmenting/Alternating Paths:} Given a matching $M$ of $G$, a path in $G$ is an {\em alternating path} for $M$ if its edges alternatively belong to $M$. An alternating path is an {\em augmenting path} for $M$ if the first and the last edges of the path do not belong to $M$.

It is well-known that if a maximal matching is nearly half the size of a maximum matching, then almost all of its edges belong to length-three augmenting paths. The following statement is folklore. For the sake of completeness, we provide a simple proof in \cref{apx:missing-proofs}.

\begin{claim}[Folklore]\label{cl:many-length-three}
    Let $M$ be a maximal matching and $M^\star$ a maximum matching. Suppose $|M| < (\frac{1}{2} + \delta) |M^\star|$. In $M \oplus M^\star$, there are at least $|M| - 4\delta|M^\star|$ length-3 augmenting paths for $M$.
\end{claim}

\paragraph{Greedy Matching:} Given a graph $G=(V, E)$ and a permutation $\pi$ of its edge-set $E$, we use $\greedyMM{G, \pi}$ to denote the greedy maximal matching obtained by iterating over the edges of $E$ in the order of $\pi$ and greedily adding each encountered edge that does not violate matching constraints to the matching.

We use the following proposition about the size of greedy matchings in vertex-subsampled subgraphs. It was first proved in \cite{BehnezhadLM20} using the techniques developed in \cite{KonradMM12}. 

\begin{proposition}[{\cite[Lemma~5.2]{BehnezhadLM20}}]\label{prop:greedy-vertex-sample}
    Let $G(V, U, E)$ be a bipartite graph, let $\pi$ be an arbitrary permutation over $E$, let $p \in (0, 1)$, and let $M$ be an arbitrary matching in $G$. Let $W$ be a subsample of $V$ including each vertex independently with probability $p$. Define $X$ to be the number of edges in $M$ whose endpoint in $V$ is matched in $\greedyMM{G[W \cup U], \pi}$; then $$\E_W[X] \geq p(|M| - 2p|V|).$$
\end{proposition}

\paragraph{Probabilistic tools:} We use the following version of Chernoff bound.
\begin{proposition}[Chernoff Bound]\label{prop:chernoff}
    Suppose $X_1, X_2, \ldots, X_n$ are independent Bernoulli random variables and $X = \sum_{i=1}^{n} X_i$. For any $t > 0$, we have
    $$\Pr[|X - \E[X]| \geq t] \leq 2\exp\left(-\frac{t^2}{3\E[X]}\right).$$
\end{proposition}

%% file: algorithm.tex
\vspace{-0.4cm}
\section{A Meta Algorithm for Beating the $\frac{1}{2}$-Approximation}

\subsection{The Algorithm}\label{sec:algorithm}

In this section, we formalize our new ``less adaptive'' meta algorithm that we informally overviewed in \cref{sec:highlevel}. We show in \cref{sec:approx} that its approximation ratio is strictly better-than-half. We later show in \cref{sec: query-process} that the size of its output matching can be estimated in sublinear time.

Before formalizing the algorithm, let us give a few useful definitions. Given an $n$-vertex  graph of maximum degree $\Delta$, and a parameter $\epsilon \in (0, .25)$, define:
\begin{whitetbox}
\begin{itemize}[itemsep=0pt,topsep=0pt,leftmargin=12pt]
    \item $p := 0.007$.
    \item $D := (c \cdot \Delta \cdot \log n)^\epsilon$ where we later fix $c \geq 1$ to be sufficiently large function of $\epsilon$.
    \item $K := 10D \log^2 n$.
    \item $\alpha_i := 1/D^{i-1}$ for $i \in [2/\epsilon]$ and $\alpha_{2/\epsilon + 1} = 0$. Note that
    $
        0 = \alpha_{2/\epsilon+1} < \alpha_{2/\epsilon} < \ldots < \alpha_2 < \alpha_1 = 1.
    $
\end{itemize}
\end{whitetbox}

We are now ready to state the algorithm, which is formalized below as \cref{alg: algorithm}.

\begin{figure}[h]
\begin{algenv}{An algorithm for beating half-approximate matching.}{alg: algorithm}
\setstretch{1.1}
\SetKwInOut{Input}{Input}
\Input{An $n$-vertex $m$-edge graph $G=(V, E)$ of max degree $\Delta$.  \textbf{Parameter:} $\epsilon \in (0, .25)$. \medskip}

Define $K$, $p$, and $\alpha_{2/\epsilon}, \ldots, \alpha_1$ as above.

Construct a sequence $T$, which for any edge $e \in E$ includes $K$ copies of $(e, \start{})$ and one copy of $(e, \extend{})$. Then random shuffle the elements in $T$.

For any vertex $v \in V$ pick a color $c_v \in \{\blue, \red\}$ uniformly and independently.\label{line:color}

Initialize $M \gets \emptyset, S \gets \emptyset$.\tcp*{Both $M$ and $S$ will be matchings of $G$.}

Initialize $M_1 \gets \emptyset, \ldots, M_{2/\epsilon} \gets \emptyset$.\tcp*{These will partition the edges in $M$.}

Draw $j^\star$ from $[2/\epsilon]$ uniformly at random.\label{line:jstar}

\For{$i =1$ to $|T|$}{
    Let $(e=\{u, v\}, X)$ be the $i$-th element in $T$.
    
    \If{$X = \start{}$ and $\deg_M(u) = \deg_M(v) = 0$}{
            Add $e$ to $M$.
            
            Add $e$ to the unique $M_i$, $i \in [2/\epsilon]$ where $\alpha_{i+1}|T| < i \leq \alpha_{i}|T|$.
            
            \If{$c_v = c_u$ or $e \not\in M_{j^\star}$}{
                Mark both $u$ and $v$ as {\em frozen}.\label{line:det-freeze}
            }
            \Else{
                With probability $1-p$ mark both $u$ and $v$ as {\em frozen}.\label{line:freeze}
            }
    }
    \If{$X = \extend{}$, $\deg_S(u) = \deg_S(v) = 0$, $\deg_M(v)+\deg_M(u) \leq 1$, $c_u \not= c_v$, and neither endpoint of $e$ is frozen\label{line:condition-add-to-S}}{
            Add $e$ to $S$.
    }
}
\Return the maximum matching in $M \cup S$.
\end{algenv}
\end{figure}

\subsection{The Approximation Guarantee}\label{sec:approx}

In this section, we prove the following approximation guarantee for \cref{alg: algorithm}. (We note that we have not attempted to optimize the constants in the statement.)

\begin{theorem}\label{thm:apx}
   Let $G$ be any $n$-vertex graph. Let $M$ and $S$ be the matchings produced by \cref{alg: algorithm} run on $G$ for parameter $\epsilon \in (0, .25)$. Then for some $\delta > 2^{-O(1/\epsilon)}$,
   $$
   \E[\mu(M \cup S)] \geq \left(\frac{1}{2} + \delta \right) \mu(G).
   $$
\end{theorem}

\subsubsection{Basic Notation and Definitions} 

For any element $\ell$ we use $\pi(\ell)$ to denote the location of $\ell$ in $T$. For any edge $e \in E$ we use \rankstart{e} (resp. \rankextend{e}) to denote the minimum $i \in [|T|]$ such that the $i$-th index of $T$ includes element $(e, \start)$ (resp. $(e, \extend)$). For any element $\ell \in T$ we say ``{\em at the time of processing $\ell$}'' to refer to the iteration of the for loop in \cref{alg: algorithm} when $i$ equals $\pi(\ell)$.

Next, we define unusual edges as follows:

\begin{definition}[unusual edges]\label{def:unusual}
We say an edge $e \in M$ is {\em unusual}, if there is some edge $e'=(u, v)$ incident to $e$ such that one of the following holds:
\begin{itemize}
    \item $\rankstart{e} < \rankextend{e'} < D \cdot \rankstart{e} $, and at the time of processing $(e', \extend{})$ the only edge in $M$ that is incident to $e'$ is $e$.
    \item $\rankextend{e'} < \rankstart{e}$, and at the time of processing $(e', \extend{})$ no edge in $M$ is incident to $e'$ in $M$.
\end{itemize}
We use $A$ to denote the subset of unusual edges of $M$.
\end{definition}

In \cref{sec:few-unusuals} we prove the following \cref{lem:few-unusuals} which shows only a small fraction of the edges in $M$ will be unusual. This is one of our key insights towards our proof of \cref{thm:apx}, and is the main place where having $K$ copies of $(e, \start{})$ compared to one copy of $(e, \extend{})$ in $T$ is used crucially.

\begin{lemma}\label{lem:few-unusuals}
    $\E|A| = o(|M|)$.
\end{lemma}

\subsubsection{The Main Argument}

In this section, we present the main building blocks of the proof of \cref{thm:apx}, deferring the proof of one key lemma (\cref{lem:num-aug-paths-conditioned}) to a later section.

Our proof of \cref{thm:apx} relies on four independent sources of randomization, which with a slight abuse of notation we denote by $j^\star$, $T$, $C$, and $F$:
\begin{itemize}
    \item $T$: The order in which \cref{alg: algorithm} processes $T$.
    \item $C$: The colors assigned to the vertices in \cref{line:color} of \cref{alg: algorithm}.
    \item $j^\star$: The index chosen in \cref{line:jstar} of \cref{alg: algorithm}.
    \item $F$: The set of edges in $M$ that get frozen in \cref{line:freeze} of \cref{alg: algorithm}.
\end{itemize}

All four sources of randomization are needed for the proof of \cref{thm:apx}. But it would be convenient to first condition on $T$ because:

\begin{observation}
    Conditioning on $T$ fully reveals the maximal matching $M$, the set $A$ of unusual edges in $M$, and all of $M_1, \ldots, M_{2/\epsilon}$.
\end{observation}

Note, however, that $S$ remains random even after conditioning on $T$ as it depends on the other three sources of randomization too.

Because $M$ is a maximal matching of $G$, we immediately get $|M| \geq \mu(G)/2$. Our plan is to show that if $M$ is only half the size of $\mu(G)$, then in expectation $S$ augments it well enough that $M$ and $S$ together include a larger matching. More formally, recall that our goal in \cref{thm:apx} is to prove that $\E[\mu(M \cup S)] \geq (\frac{1}{2} + \delta) \mu(G)$. This clearly holds if $|M| \geq (\frac{1}{2} + \delta) \mu(G)$. So let us for the rest of the proof assume that $M$ is smaller.

\begin{assumption}\label{ass:M-small}
    $|M| < (\frac{1}{2} + \delta) \mu(G)$.
\end{assumption}

Plugging this assumption into \cref{cl:many-length-three} gives:

\begin{observation}\label{obs:many-length-three}
    At least $|M| - 4\delta|M^\star|$ edges of $M$ belong to length-3 augmenting paths in $M \oplus M^\star$.
\end{observation}

Our next claim shows that there is one subset $M_j$ of $M$ that has several nice properties. We will later argue that $M_j$ will be well augmented by $S$ in expectation.

\begin{claim}\label{cl:first-significant-Mi}
    Define $q(x) := 2^{20(x-3/\epsilon)}$. There is $M_j$ such that all the following hold:
    \begin{enumerate}
        \item $|M_j| \geq q(j) |M|$,
        \item $|M_1| + \ldots + |M_{j-1}| \leq 2^{-19} |M_j|$,
        \item For any two edges $e, e' \in M_j$, $\rankstart{e'} / D \leq \rankstart{e} \leq \rankstart{e'} \cdot D$.
    \end{enumerate}
\end{claim}
\begin{proof}
    First, there should exist $j \in [2/\epsilon]$ such that $|M_j| \geq q(j) |M|$ as otherwise
    $$
        |M_1|+\ldots + |M_{2/\epsilon}| < \sum_{i=1}^{2/\epsilon} q(i) |M| = \frac{2^{20} + \ldots + 2^{40/\epsilon}}{2^{60/\epsilon}} |M| < |M|,
    $$
    which contradicts the fact that $M_i$'s partition $M$. Take the smallest $j$ with $|M_j| \geq q(j)|M|$, noting that the first property of the lemma is satisfied for $M_j$. We have
    $$
        |M_1| + \ldots + |M_{j-1}| < \sum_{i=1}^{j-1} q(i) |M| < \frac{2^{20} + \ldots + 2^{20(j-1)}}{2^{60/\epsilon}}|M| < 2^{20j-19-60/\epsilon} |M| < 2^{-19} |M_j|,
    $$
    so the second property also holds for $M_j$.
    
    For the third property, note from the definition of $\alpha_1, \ldots, \alpha_{2/\epsilon+1}$ that if $j \not= \frac{2}{\epsilon}$, then all edges $e$ in $M_j$ have the same $\rankstart{e}$ up to a factor of $D$. So it suffices to show $j \not = \frac{2}{\epsilon}$. Observe that for every $e \in M_{2/\epsilon}$, by definition we have
    $$
    \rankstart{e} \leq
    \alpha_{\frac{2}{\epsilon}}|T| =
    \frac{|T|}{D^{2/\epsilon - 1}} \leq
    \frac{2mK}{D^{2/\epsilon - 1}} =
    \frac{20m \log n}{D^{2/\epsilon - 2}} =
    \frac{20m \log n}{(c \Delta \log n)^{\epsilon(2/\epsilon - 2)}} <
    \frac{20m \log n}{c \Delta \log n} = \frac{20 m }{c \Delta}.
    $$
    Now by choosing $c$ to be a sufficiently large function of $\epsilon$, we can further guarantee that
    $$
        \rankstart{e} < \frac{q(1)m}{4\Delta} \leq q(1) |M|,
    $$
    where the last inequality follows because\footnote{Edge color the graph greedily using $2\Delta$ colors and pick the largest color class which will be a matching.} $\mu(G) \geq \frac{m}{2\Delta}$ and $|M| \geq \mu(G)/2$. Since there are less than $q(1)|M|$ edges $e$ for which $\rankstart{e} < q(1)|M|$, we get $|M_{2/\epsilon}| < q(1)|M|$. Combined with the first property of the claim that $|M_j| \geq q(j) |M|$ and given that $q(x) \geq q(1)$ for every $x \geq 1$, we get that $j \not= 2/\epsilon$, implying the third property and completing the proof.
\end{proof}

Now consider the set $P_j$ all length three augmenting paths in $M^\star \oplus M$ where the middle edge belongs to $M_j$ and that the vertices along these paths are alternatively \blue{} and \red{} as follows:
\begin{flalign*}
    &P_j := \left\{ (x, u, v, y) \,\,\Bigg\vert \,\, 
    \parbox{11cm}{
    $(x, u, v, y)$ is an augmenting path for $M$, $(x, u) \in M^\star, (v, y) \in M^\star,\\
    (u, v) \in M_j$, $c_x = \red{}$, $c_u = \blue$, $c_v = \red{}$, $c_y = \blue{}$
    }
    \right\}.
\end{flalign*}

\smallskip\smallskip

\begin{observation}
    Conditioning on $T$ and $C$ fully reveals $P_j$.
\end{observation}
\begin{myproof}
    Definition $P_j$ only depends on the vertex colors which are determined by $C$, and matchings $M_j$ and $M$ which are fully determined by $T$.
\end{myproof}

The following lemma, which conditions on everything except $F$, is the key to \cref{thm:apx}. We defer its proof to \cref{sec:proof-aug-path}.

\begin{lemma}\label{lem:num-aug-paths-conditioned}
    Let us condition on $T$, $C$, $j^\star = j$, and let $P_j$ and $M_j$ be as above. Let $Y$ denote the number of length three augmenting paths for $M$ in $M \oplus S$. Then
    $$
        \E_F\Big[Y \mid T, C, j^\star = j \Big] \geq p|P_j| - \left(4p^2 + 2^{-18}\right)|M_j| - 12|A|.
    $$
\end{lemma}

Let us first see how \cref{lem:num-aug-paths-conditioned} implies \cref{thm:apx}.

\begin{proof}[Proof of \cref{thm:apx}]
    We assume \cref{ass:M-small} holds, or otherwise the theorem is trivial. Recall from  \cref{obs:many-length-three} that at most $4\delta\mu(G)$ edges of $M$ (and thus $M_j$) are not in length-three augmenting paths in $M \oplus M^\star$. Since the colors are random and independent, each of these length-three augmenting paths is colored in the way specified in the definition of $P_j$ with probability exactly $1/2^4$. Hence,
    \begin{equation}\label{eq:xguh-9810237}
        \E_{C}[|P_j|] \geq \frac{1}{2^4}(|M_j| - 4\delta\mu(G)).
    \end{equation}
    
    \noindent Additionally, recall from \cref{cl:first-significant-Mi} that 
    \begin{equation}\label{eq:uheh-27139733}
    |M_j| \geq q(j)|M| \geq q(1)|M| \geq \frac{q(1)}{2} \mu(G).
    \end{equation}
    
    \noindent Also recall from \cref{lem:few-unusuals} that 
    \begin{equation}\label{eq:rpgts-123897}
        \E_T[|A|] = o(\mu(G)).
    \end{equation}
    
    \noindent Taking expectation over $C$ and $T$ from both sides of the inequality of \cref{lem:num-aug-paths-conditioned}, we get
    \begin{flalign*}
        \E_{F, C, T}[Y \mid j^\star = j] &\geq \E_{C, T}\left[ p|P_j| -   \left(4p^2 + 2^{-18}\right)|M_j| - 12|A| \right]\\
        &\geq \frac{p}{16} (|M_j| - 4\delta \mu(G)) - \left(4p^2 + 2^{-18}\right)|M_j| - o(\mu(G)) \tag{By \Cref{eq:xguh-9810237} and \Cref{eq:rpgts-123897}}\\
        &= \left( \frac{p}{16} - 4p^2 - 2^{-18} \right) |M_j| - \frac{p\delta}{4} \mu(G) - o(\mu(G))\\
        &\geq \left( \left(\frac{p}{16} - 4p^2 - 2^{-18} \right) \frac{q(1)}{2} - \frac{p \delta}{4} - o(1) \right) \mu(G)\tag{By \Cref{eq:uheh-27139733}.}\\
        &> \left( 10^{-4} q(1) -  0.002 \cdot \delta - o(1) \right) \mu(G) \tag{Since $p = 0.007$.}\\
        &> \frac{2\delta}{\epsilon} \mu(G).\tag{Since $q(1) = 2^{20-60/\epsilon}$, $\delta = 2^{-70/\epsilon}$.}
    \end{flalign*}
    
    \noindent This, in turn, implies that
    $$
        \E_{F, C, T, j^\star}[Y] \geq \Pr[j^\star = j] \cdot \E_{F, C, T}[Y \mid j^\star = j] \geq \frac{\epsilon}{2} \cdot \frac{2\delta}{\epsilon} \mu(G) = \delta \mu(G).
    $$
    
    \noindent Since $Y$ is a lower bound on the number of length three augmenting paths for $M$ in $M \oplus S$, we get 
    $$
        \E_{F, C, T, j^\star}[\mu(M \cup S)] \geq |M| + \E_{F, C, T, j^\star}[Y] \geq \left(\frac{1}{2} + \delta \right) \mu(G),
    $$
    which is the desired bound.
\end{proof}

\subsubsection{Proof of \cref{lem:num-aug-paths-conditioned}}\label{sec:proof-aug-path}

\begin{proof}
    For brevity, we use $\E'[X]$ as a shorthand for $\E_F[X \mid T, C, j^\star = j]$ throughout the proof.
    
    Recall that in \cref{alg: algorithm}, when visiting an element $((u, v), \extend)$ in $T$, we add it to $S$ if all the following conditions hold, where we have deliberately broken the last condition in \cref{line:condition-add-to-S} of \cref{alg: algorithm} into two sub-conditions \ref{prop-papx:4} and \ref{prop-papx:5}: 
    \begin{enumerate}[label=(C$\arabic*$)]
        \item $\deg_S(u) = \deg_S(v) = 0$,\label{prop-papx:1}
        \item $\deg_M(v) + \deg_M(u) \leq 1$,\label{prop-papx:2}
        \item $c_u \not= c_v$,\label{prop-papx:3}
        \item neither of $u$ or $v$ is frozen in \cref{line:det-freeze} of \cref{alg: algorithm}.\label{prop-papx:4}
        \item neither of $u$ or $v$ is frozen in \cref{line:freeze} of \cref{alg: algorithm}.\label{prop-papx:5}
    \end{enumerate}
    
    Since we have conditioned on $T$ and $C$, both the matching $M$ and the vertex colors are fully revealed. Thus, every element $((u, v), \extend) \in T$ which upon being visited violates one of the conditions \ref{prop-papx:2}, \ref{prop-papx:3}, or \ref{prop-papx:4} is fully revealed a priori and can be discarded. On the flip side, condition \ref{prop-papx:5} depends on $F$ and so remains random.
    
    Now suppose for the sake of the analysis that we also discard all elements $(e, \extend)$ where $e$ is incident to an unusual edge in $M$. We emphasize that discarding these elements might change matching $S$, but we will show later that this effect is not significant.
    
    \paragraph{Useful definitions:} Next, we give a few useful definitions that we use in the proof. First, define 
    \begin{flalign*}
    \hat{B} &:= \{ v \mid \exists (u, v) \in M_j, c_v = \blue, c_u = \red \}\\
    \hat{R} &:= \{ u \mid \exists (u, v) \in M_j, c_v = \blue, c_u = \red \}.
    \end{flalign*}
    Also let $U$ be the set of vertices that are left unmatched by $M_j, \ldots, M_{2/\epsilon}$, and define
    \begin{flalign*}
        U_R := \{ u \mid u \in U, c_u = \red{} \}, \qquad U_B := \{ v \mid v \in U, c_v = \blue\}.
    \end{flalign*}
    Moreover, define matchings $M_B$ and $M_R$ as
    \begin{flalign*}
        M_B := \{ (x, u) \mid \exists (x, u, \cdot, \cdot) \in P_j \}, \qquad M_R := \{ (y, v) \mid \exists (\cdot, \cdot, v, y) \in P_j \}.
    \end{flalign*}
    Finally, define $H_B$ (resp. $H_R$) to be the bipartite graph with vertex parts $\hat{B}$ and $U_R$ (resp. $\hat{R}$ and $U_B$), including an edge $e$ of $G$ between its vertex parts iff $(e, \extend{})$ is not discarded.
    
    It can be confirmed from the definitions above that the sets $\hat{B}$, $\hat{R}$, $U_R$, $U_B$ are all disjoint. This implies that $H_B$ and $H_R$ are vertex disjoint. The next observation also follows immediately from the definitions above.
    
    \begin{observation}\label{obs:MBRinHBR}
        All edges of $M_B$ (resp. $M_R$) belong to $H_B$ (resp. $H_R$).
    \end{observation}
    
    See \cref{fig:apx} for an illustration of some of these definitions.
    
    We show that $H_B$ and $H_R$ include all the remaining \extend{} elements.
    
    \begin{figure}
        \centering
        \includegraphics[scale=0.6]{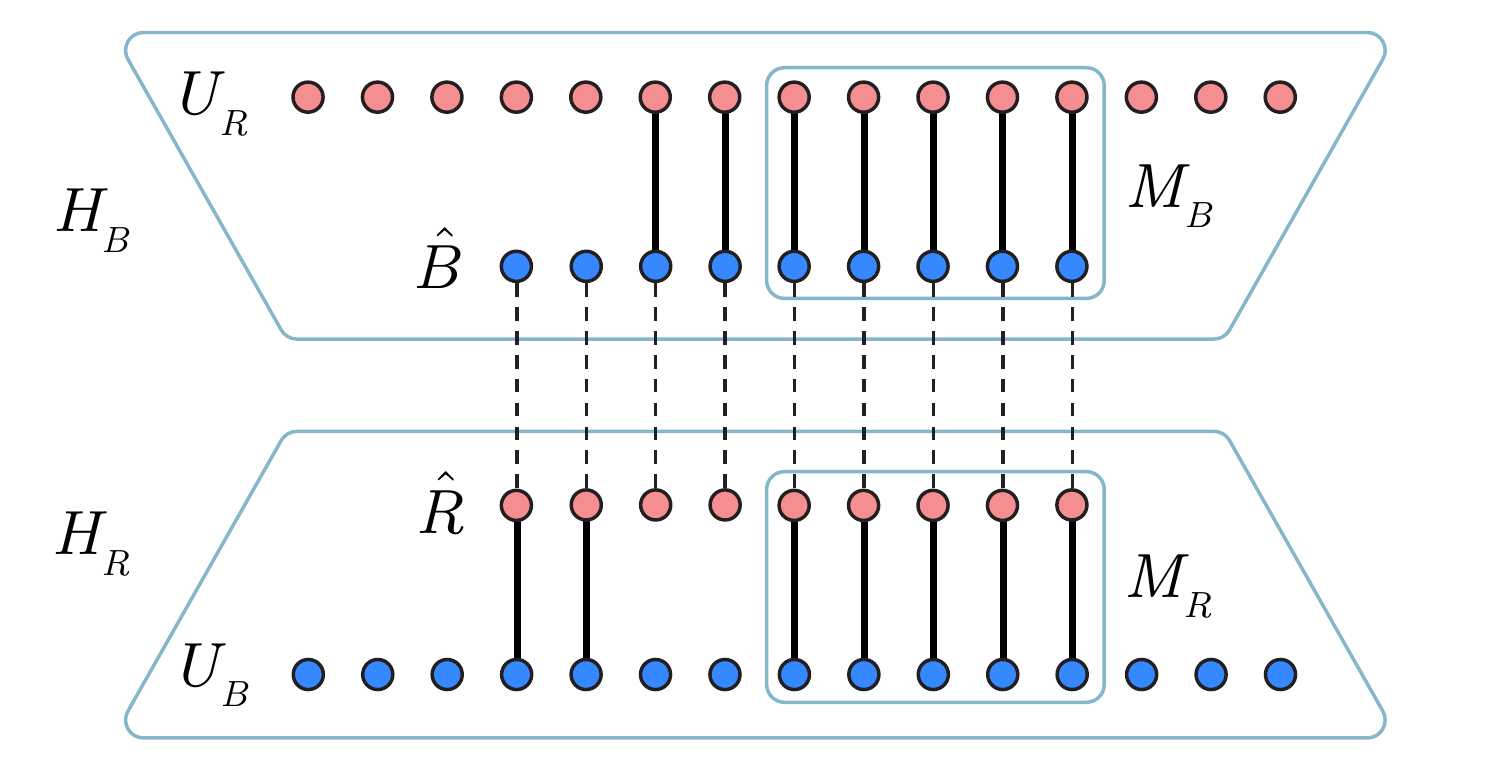}
        \caption{Illustration of $\hat{B}$, $\hat{R}$, $U_B$, $U_R$, $H_B$, $H_R$, and matchings $M_B$, $M_R$. The solid edges are the edges of $M^\star$, and the dashed edges are the edges of $M_j$ whose endpoints have different colors.
        }
        \label{fig:apx}
    \end{figure}
    
    \begin{claim}\label{cl:bhgr-28980231}
        For every element $(e, \extend)$ that is not discarded, $e$ belongs to $H_B$ or $H_R$.
    \end{claim}
    \begin{myproof}
        Take an undiscarded element $(e=(u, v), \extend)$. Take the edge $e'=(u, w) \in M$ with the smallest $\rankstart{e'}$ that is incident to $e$. Note that such $e'$ should exist because $M$ is a maximal matching of $G$. There are three possible cases, and only the last one does not lead to a contradiction:
        \begin{itemize}
            \item \textbf{Case 1 --} $\rankextend{e} < \rankstart{e'}$: In this case, $e'$ is unusual by the second condition of \cref{def:unusual} and so $(e, \extend)$ must be discarded, a contradiction.
            \item \textbf{Case 2 --} $\rankextend{e} \geq \rankstart{e'}$ and ($e' \not\in M_{j^\star}$ or $c_{u} = c_{w}$): In this case, the endpoints of $e'$ must be frozen in \cref{line:det-freeze} of \cref{alg: algorithm}. So condition \ref{prop-papx:4} does not hold for $e$ when processing $(e, \extend)$ and we should discard it, a contradiction.
            \item \textbf{Case 3 --} $\rankextend{e} \geq \rankstart{e'}$ and $e' \in M_{j^\star}$ and $c_{u} \not= c_{w}$: This is the only case that does not lead to a contradiction.
        \end{itemize}
        The condition of Case 3 immediately implies that if $c_{u} = \blue$ then $u \in \hat{B}$ otherwise $u \in \hat{R}$. 
        
        Next, we show that $v$ must be unmatched by $M_{j^\star}, \ldots, M_{2/\epsilon}$, and so $v \in U$. First, note that if $v$ is also matched by $M$ through an edge $e''$, then by definition of $e'$ it must satisfy $\rankstart{e''} > \rankstart{e'}$. Since $e' \in M_{j^\star}$ and the ranks of the edges of $M_{j^\star+1}, \ldots, M_{2/\epsilon}$ are all smaller than those in $M_{j^\star}$, this implies $e'' \not\in M_{j^\star+1}, \ldots, M_{2/\epsilon}$. 
        So it remains to show $e'' \not\in M_{j^\star}$. To see this, note that if $e'' \in M_{j^\star} = M_j$, then from \cref{cl:first-significant-Mi} part 3, we get $\rankstart{e''} < \rankstart{e'} \cdot D$. Now if $\rankextend{e} > \rankstart{e''}$, then at the time of processing $(e, \extend)$ both $e'$ and $e''$ are in $M$ and we have $\deg_M(u) + \deg_M(v) = 2$, which means $(e, \extend)$ violates \ref{prop-papx:2} and must be discarded, a contradiction. So we should have $\rankextend{e} < \rankstart{e''} < \rankstart{e'} \cdot D$. This is again a contradiction because by the first condition of \cref{def:unusual}, $e'$ must be unusual, and so $e$ must be discarded.
        
        Finally, note that since $(e, \extend)$ is not discarded, it should satisfy \ref{prop-papx:3} and so $c_u \not= c_v$. Combined with the discussion above this means that if $u \in \hat{B}$ then $v \in U_R$ and if $u \in \hat{R}$ then $v \in U_B$. Hence $e$ must belong to one of $H_B$ and $H_R$, completing the proof.
    \end{myproof}
    
    Now consider the construction of $S$ from the remaining undiscarded elements. We iterate over these elements in the order specified by $T$, and whenever we see an element $(e, \extend)$ that satisfies all of \ref{prop-papx:1}--\ref{prop-papx:5} we add it to $S$. As discussed, conditions \ref{prop-papx:2}--\ref{prop-papx:4} are automatically satisfied by all undiscarded  elements, which only leaves \ref{prop-papx:1} and \ref{prop-papx:5}. Condition \ref{prop-papx:1} is simply the greedy matching constraint. Condition \ref{prop-papx:5} depends on the randomization in $F$. An edge in $H_B$ (resp. $H_R$) satifies \ref{prop-papx:5} if its endpoint in $\hat{B}$ (resp. $\hat{R}$) is not frozen. An important observation is that each vertex in $\hat{B}$ (resp. $\hat{R}$) is frozen independently from the other vertices of $\hat{B}$ (resp. $\hat{R}$) with probability $1-p$. (We emphasize though that these decisions are not mutually independent when we consider the vertices of both $\hat{B}$ and $\hat{R}$ together.)

    Let $S_B$ and $S_R$ be the subset of edges of $S$ that respectively belong to $H_B$ and $H_R$. Our goal is to apply \cref{prop:greedy-vertex-sample} on both $H_B$ and $H_R$.

    First, we apply \cref{prop:greedy-vertex-sample} by letting $G = H_B$, $V = \hat{B}$, $U = U_R$, the subsample $W$ being the subset of vertices in $\hat{B}$ that are not frozen, and $M = M_B$ (recalling from \cref{obs:MBRinHBR} that $M_B$ is completely inside $H_B$). Using $X_B$ to denote the set of vertices in $V(M_B) \cap \hat{B}$ that get matched in $S$ to $U_R$, we get from \cref{prop:greedy-vertex-sample} that:
    $$
        \E'[|X_B|] \geq p(|M_B|-2p|\hat{B}|) \geq p|P_j| - 2p^2|M_j|.
    $$
    
    Next, we apply \cref{prop:greedy-vertex-sample} by letting $G = H_R$, $V = \hat{R}$, $U = U_B$, the subsample $W$ being the subset of vertices in $\hat{R}$ that are not frozen, and $M = M_R$ (recalling from \cref{obs:MBRinHBR} that $M_R$ is completely inside $H_R$). Using $X_R$ to denote the set of vertices in $V(M_R) \cap \hat{R}$ that get matched in $S$ to $U_B$, we get from \cref{prop:greedy-vertex-sample} that:
    $$
        \E'[|X_R|] \geq p(|M_R|-2p|\hat{R}|) \geq p|P_j| - 2p^2|M_j|.
    $$
    
    \paragraph{Adding back the discarded {\normalfont \extend{}} elements:} We now add back the \extend{} elements incident to unusual edges that we discarded earlier. To do so, we iteratively take an arbitrary vertex of an arbitrary unusual edge in $M$, and add back all the \extend{} elements incident to it that we discarded, and re-compute matching $S$. 
    
    \begin{claim}\label{cl:changes-two}
        Let $S_1$ and $S_2$ be the edges in matching $S$ before and after adding back the discarded edges of a vertex $v$. At most two vertices can be matched in one of $S_1, S_2$ bot not the other.
    \end{claim}
    \begin{myproof}
        Since $S_1$ and $S_2$ are both matchings, $S_1 \Delta S_2$ is a collection of paths and cycles. Thus any vertex whose matching-status differs in $S_1$ and $S_2$ must be an endpoint of a path in $S_1 \Delta S_2$.  Suppose for contradiction that there are more than two such vertices. Then we have more than one path in $S_1 \Delta S_2$. Take the lowest rank edge $e$ in the path that does not include $v$. It can be confirmed that whether the conditions \ref{prop-papx:1}--\ref{prop-papx:5} are satisfied for $e$ remains the same in both $S_1$ and $S_2$, so either $e$ belongs to both or neither, contradicting that $e \in S_1 \Delta S_2$.
    \end{myproof}
    
    Let us now define $X'_B$ and $X'_R$ to be the analogs of $X_B$ and $X_R$ after we add back the discarded edges incident to unusual edges of $M$. More precisely, let $X'_B$ (resp. $X'_R$)  denote the set of vertices in $V(M_B) \cap \hat{B}$ (resp. $V(M_R) \cap \hat{R}$) that are matched in $S$ to $U_R$ (resp. $U_B$).
    
    Note that a vertex $v$ may belong to $X_B \setminus X'_B$ for two reasons: either $v$ was matched in $S$ before adding back the discarded edges but then got unmatched after doing so, or $v$ remains matched in $S$, but to a vertex not in $U_R$. \cref{cl:changes-two} bounds the total number of vertices of the former type by $2 \times 2|A|=4|A|$. For the latter type,  note from \cref{cl:bhgr-28980231} that any edge of $v$ that is not discarded goes to $U_R$. So the new match of $v$ after adding the discarded edges must be a discarded edge. But each discarded edge that belongs to $S$ must match one endpoint of one of the $|A|$ unusual edges in $M$, so the total number of such edges is no more than $2|A|$. Thus, overall
    \begin{equation}\label{eq:hbxg-928}
        \E'[|X'_B|] \geq \E'[|X_B|] - 6|A| \geq p|P_j| - 2p^2|M_j| - 6|A|.
    \end{equation}
    Applying the same argument on $X'_R$ gives
    \begin{equation}\label{eq:hbyg-213}
        \E'[|X'_R|] \geq \E'[|X_R|] - 6|A| \geq p|P_j| - 2p^2|M_j| - 6|A|.
    \end{equation}
    
    Now, let $UF_j$ denote the set of $M_j$ edges of $P_j$ that are unfrozen. Each edge $e \in UF_j$ has one endpoint colored \blue{} and one that is colored \red{} by definition of $P_j$. The set of \blue{} (resp. \red{}) endpoints of $UF_j$ that are matched in $S$ to a vertex in $U_R$ (resp. $U_B$) is exactly $X'_B$ (resp. $X'_R$). We have:
    \begin{flalign}
        \nonumber \E'[\# \text{ of } (u, v) \in UF_j \text{ s.t. } u \in X'_B, v \in X'_R] &\geq \E'[|UF_j| - (|UF_j| - |X'_B|) - (|UF_j| - |X'_R|)]\\
        \nonumber &= \E'[|X'_B|] + \E'[|X'_R|]-\E'[|UF_j|]\\
        &\geq p|P_j| - 4p^2|M_j| - 12|A|.\label{eq:hclx-9138}
    \end{flalign}
    The last inequality follows from \cref{eq:hbxg-928} and \cref{eq:hbyg-213}, and the fact that $\E[|UF_j|] = p|P_j|$ because $P_j$ has $|P_j|$ edges in $M_j$ and each one is unfrozen with probability $p$.
    
    To finish the proof, note that if for an edge $(u, v) \in UF_j$ we have $u \in X'_B$ and $v \in X'_R$, then $S$ matches both $u$ and $v$ to vertices that are unmatched by $M_j, \ldots, M_{2/\epsilon}$. Therefore, only if these vertices are also unmatched by $M_1, \ldots, M_{j-1}$ we have a length three augmenting path. Therefore,
    \begin{flalign*}
        \E'[Y] &\geq \E'[\# \text{ of } (u, v) \in UF_j \text{ s.t. } u \in X'_B, v \in X'_R] - 2\sum_{i=1}^{j-1} |M_i|\\
        &\geq p|P_j| - 4p^2|M_j| - 12|A| - 2 \times \frac{1}{1000}|M_j| \tag{By \cref{eq:hclx-9138} and \cref{cl:first-significant-Mi}.}\\
        &=p|P_j| - \left(4p^2 + \frac{1}{500}\right)|M_j| - 12|A|.
    \end{flalign*}
    This finishes the proof of \cref{lem:num-aug-paths-conditioned}
\end{proof}

\subsubsection{Bounding Unusual Edges}\label{sec:few-unusuals}

In this section we prove \cref{lem:few-unusuals} that $\E|A| = o(|M|)$.

We start by proving the following auxiliary claim.

\begin{claim}\label{clm: start-soon}
    With probability $1 - 1/n^2$, every edge $e \in E$ satisfies $\rankstart{e} \leq 8m\log n$. 
\end{claim}
\begin{proof}
    Let $z := 8m\log n$. We show that with probability $1-1/n^2$ every edge $e \in E$ satisfies $\rankstart{e} \leq z$.  This statement is trivial if $|T| \leq z$ so assume $|T| > z$. Fix an arbitrary edge $e \in E$. For the event $\rankstart{e} > z$ to happen, all $K$ copies of $(e, \start)$ should appear after the first $z$ elements in $T$. Thus, noting that $|T| = m(K+1)$, $z = 8m\log n$, and $K \geq 1$, we have
    $$
        \Pr[\rankstart{e} \leq z] \geq 1 - \left(1-\frac{z}{|T|}\right)^{K} = 1 - \left(1-\frac{8m\log n}{m (K+1)}\right)^{K} \geq 1 - e^{-\frac{8K \log n}{K+1}} \geq 1 - n^{-4}.
    $$
    A union bound over all choices of $e$, which there are less than $n^2$ many, proves our first claim.
\end{proof}

We are now ready to prove \cref{lem:few-unusuals}.

\begin{proof}[Proof of \cref{lem:few-unusuals}]
We count the number of unusual edges separately based on the two cases in \Cref{def:unusual}. Consider the second case. Let $A'$ be the set of $(e', \extend{})$ elements in $T$ such that at the time of processing $(e', \extend{})$, both endpoints of $e'$ are unmatched in $M$. If we show that $\E_\pi|A'| = o(|M|)$, since each edge in $A'$ has at most two incident edges in $M$, the number of unusual edges in the second case of \Cref{def:unusual} will be $o(|M|)$. Consider the following equivalent construction of $A'$. We iterate over $T$, processing its elements one by one. If we see an element $(e, \start{})$ we add $e$ to $M$ and remove all the unprocessed elements $(e', X)$, $X \in \{\extend, \start\}$ from $T$ such that $e'$ shares an endpoint with $e$. If we see an element $(e, \extend)$, we add it to $A'$.

At any time during this process, the remaining \start{} elements are at least $K$ times more than the \extend{} element since for each \extend{} element $(e, \extend)$ that remains in $T$, all $K$ copies of $(e, \start)$ must also remain in $T$. This means that every time we reveal the next remaining element of $T$, it is an \extend{} element with probability at most $1/(K+1)$. Furthermore, there are at most $\mu(G)$ steps where we process a \start{} element since each time we add an edge to $M$ and clearly $|M| \leq \mu(G)$. This implies that $\E|A'| \leq \frac{\mu(G)}{K}$. Combining $\mu(G) \leq 2|M|$ and $K = \Omega(\log n)$, we have $\E|A'| = o(|M|)$.

Now consider the first case in \Cref{def:unusual}. Define $\sigma: T \rightarrow \mathbb{R}$ as $\sigma((e, \start{})) := \pi((e, \start{}))$ and $\sigma((e, \extend{})) := \pi((e, \extend{})) / D$. Similar to the proof of the second case, we process the elements in $T$ one by one. However, instead of $\pi$, we process them in the increasing order of their $\sigma$ values. Let $A''$ be the set of \extend{} elements in $T$ such that at the time of processing $(e', \extend{})$ both endpoints of $e'$ are unmatched in $M$. Note that if an edge $e$ is unusual because of the first case of \Cref{def:unusual}, there exists an element $(e', \extend{})$ such that $\pi_{\extend{}}(e')/D < \pi_{\start{}}(e)$ and at the time of processing $(e',\extend{})$, the only incident edge to $e'$ in $M$ is $e$. Hence, we must process $(e', \extend{})$ before $(e, \start{})$ according to new ordering of elements and at the time of processing $(e', \extend{})$, there is no incident edge to $e'$ in $M$. Therefore, $|A''|$ is an upperbound for number of unusual edges in first case of \Cref{def:unusual}. Similar to the previous case, it suffices to show $\E_\pi|A''| = o(|M|)$ since each edge in $A''$ has at most two incident edges in $M$. Consider again the following equivalent construction of $A''$, where we iterate over the elements based on $\sigma$ one by one. If we see an element $(e, \start{})$ we add $e$ to $M$ and remove all the unprocessed elements $(e', X)$, $X \in \{\extend, \start\}$ such that $e'$ shares an endpoint with $e$. If we see an element $(e, \extend)$, we add it to $A''$.

Let $E_S$ be the event that $\rankstart{e} \leq 8 m\log n$ for all edges $e \in E$, and let $\bar{E_S}$ be the complement event. Noting from \cref{clm: start-soon} that $\Pr[\bar{E_S}] \leq 1/n^2$. We get
\begin{flalign*}
\Pr[\sigma((e, \extend{})) \leq \rankstart{e}] &= \Pr\left[\rankextend{e}\leq D \cdot \rankstart{e}\right]\\
&\leq \Pr[\bar{E_S}] + \Pr\left[\rankextend{e} \leq D \cdot \rankstart{e} \mid E_S \right]\\
&\leq \frac{1}{n^2} + \frac{8mD\log n}{m(K+1)} \\
&\leq \frac{8D\log n}{K}.
\end{flalign*}

Hence, the probability of the element that we are processing to be an \extend{} element is at most $\frac{8D\log n}{K}$. Similar to the former case, since there are at most $\mu(G)$ steps in the process that an \start{} element is added to $M$, we have $\E|A''| \leq \frac{\mu(G)\cdot 16D\log n}{K}$. Combining with $\mu(G) \leq 2|M|$, $D=(c \cdot \Delta \cdot \log n)^\epsilon$, and $K = 10D \log^2 n$, implies $\E|A''| = o(|M|)$ which completes the proof.
\end{proof}

%% file: query_process.tex

\section{A Local Query Algorithm and its Complexity}\label{sec: query-process}

In this section, we present a local query process that determines whether a given vertex has any edge in matchings $M$ and $S$ of \cref{alg: algorithm} without constructing the whole output of \cref{alg: algorithm}.

While all the randomizations in \cref{alg: algorithm} were crucial for the approximation guarantee of \cref{sec:approx} to hold, the bounds of this section only rely on the random shuffling of the sequence $T$. In particular, the result of this section continues to hold even if all the other coin flips of \cref{alg: algorithm} besides the order of $T$ are picked adversarially.

To help the discussion above simplify the presentation of this section, we regard the index $j^\star$ and the color $c_v$ of any vertex $v$ as given. Alternatively, one could pick the color $c_v$ of any vertex uniformly at random whenever we access $c_v$. 

 Note from \cref{alg: algorithm} that whether a vertex is frozen depends on its edge in $M$, if any. In particular, a vertex $v$ is frozen iff it is matched in $M$ and its match is also frozen. This can essentially be seen as freezing the edges of $M$ instead of the vertices, which will be the more convenient view for our purpose in this section. More generally, instead of just the \start{} elements that end up in $M$, we define (\cref{def: freeze}) whether a \start{} element of $T$ is frozen or not. This way, we say a vertex $v$ is frozen iff there is an edge $e$ incident to $v$ such that $e \in M$ and the corresponding $(e, \start)$ element in $T$ that adds $e$ to $M$ is frozen.

\begin{definition}\label{def: freeze}
We say an element $\ell = ((u, v), \start{}) \in T$ is {\em frozen} if either of the following conditions holds:
\begin{itemize}[itemsep=0pt,topsep=0pt]
    \item $c_u = c_v$,
    \item $\pi(\ell) \leq \alpha_{j^\star+1}|T|$ or $\pi(\ell) > \alpha_{j^\star}|T|$,
    \item The corresponding $\textup{Bernoulli}(1 - p)$ variable in \cref{line:freeze} of \cref{alg: algorithm} is one. 
\end{itemize}
\end{definition}

We emphasize again that the randomization in the last bullet of \cref{def: freeze} is only needed for the approximation guarantee, and can be regarded as (adversarially) fixed for our purpose in this section.

Now let $\pi$ be a permutation of the sequence $T$ consisting of edge copies in graph $G$. We write $\MS(G, \pi)$ to denote the subgraph $M \cup S$ constructed by \cref{alg: algorithm} when processing $T$ in the order of $\pi$. The argument $\pi$ in $\MS(G, \pi)$ is meant to emphasize that once we feed $\pi$ into \cref{alg: algorithm}, its output is uniquely determined as we have fixed the other sources of randomization.

Having discussed our basic analysis setup, our local query process for determining whether a given vertex $v$ is matched in $M$ or $S$ is formalized below as \cref{alg: vertex_oracle}. The algorithm calls two other edge subroutines formalized as \cref{alg: start_oracle,alg: extend_oracle}. Subroutine \cref{alg: start_oracle} determines if a \start{} element is matched in $M$ by recursively calling the subroutine for incident \start{} elements with a lower rank. Similarly, subroutine \cref{alg: extend_oracle} determines if an \extend{} element is matched in $S$ by recursively calling the subroutines for incident \start{} and \extend{} elements with lower ranks.

Let $F(v, \pi)$ denote the total number of recursive calls to the edge oracles of \cref{alg: start_oracle,alg: extend_oracle} during the execution of $\VO(v, \pi)$. The following theorem is the main result of this section.

\begin{theorem}\label{thm: query-complexity}
For a randomly chosen vertex $v$ and a random permutation $\pi$ of $T$,
\begin{align*}
    \E_{v, \pi}[F(v,\pi)] = O(K\bar{d}\cdot \log^4 n),
\end{align*}
where $\bar{d}$ is the average degree of $G$.
\end{theorem}

\begin{figure}[p]

\begin{algenv}{The vertex oracle $\VO(u,\pi)$ which determines the matching-status of $u$ in the outputs $M$ and $S$ according to permutation $\pi$.}{alg: vertex_oracle}
    Let $\ell_1 = (e_1, X_1), \ldots, \ell_r = (e_r, X_r)$ be the elements in $T$ such that $e_i = (u, v_i)$, and $\pi(\ell_1) < \ldots < \pi(\ell_r)$. 

    $ST \gets \false$, $EX \gets \false$.
    
    \For{$i$ in $1 \ldots r$}{
    
    \IIf{$X_i = \extend{}$ and $c_u = c_{v_i}$}{\Continue.}
    
    \If{$ST = \false$ and $X_i = \start$ and $\EOS(\ell_i, v_i, \pi) = \true$}{
            $ST \gets \true$
    }
    \If{$EX = \false$ and $X_i = \extend$ and $\EOE(\ell_i, v_i, \pi, ST) = \true$}{
            $EX \gets \true$
    }
    }
    
    \Return $ST$, $EX$.
\end{algenv}

\begin{algenv}{
The edge oracle $\EOS(\ell=(e, \start{}),u,\pi)$ for \start{} elements. Here $u$ in the input must be an endpoint of $e$.
}{alg: start_oracle}
    \IIf{\textup{EOS}$((e, \start{}), u, \pi)$ is already computed}{\Return the computed answer.} \label{line: caching}

    Let $\ell_1 = (e_1, X_1), \ldots, \ell_r = (e_r, X_r)$ be the elements in $T$ such that $e_i = (u, v_i)$, $X_i= \start{}$ for all $i$, and $\pi(\ell_1) < \ldots < \pi(\ell_r) < \pi(\ell)$.
    
        \For{$i$ in $1 \ldots r$}{
            \IIf{$\EOS(\ell_i, v_i, \pi) = \true$}{\Return \false.}
        }
    \Return \true.
\end{algenv}

\begin{algenv}{
The edge oracle $\EOE(\ell=(e, \extend{}),u,\pi, ST_w)$ for \extend{} elements. Here $e=(u,w)$ and $ST_w$ indicates whether vertex $w$ is matched by $M$ in $\MS(G, \pi)$ through an element of rank smaller than $\pi(\ell)$.
}{alg: extend_oracle}
    \IIf{\textup{EOE}$((e, \extend{}), u, \pi, ST_w)$ is already computed}{\Return the computed answer.}
    
Let $\ell_1 = (e_1, X_1), \ldots, \ell_r = (e_r, X_r)$ be the elements in $T$ such that $e_i = (u, v_i)$, and $\pi(\ell_1) < \ldots < \pi(\ell_r) < \pi(\ell)$.

    $ST_u \gets \false$
    
    \For{$i$ in $1 \ldots r$}{
        \IIf{$X_i = \extend{}$ and $c_u = c_{v_i}$}{\Continue.}

        \If{$ST_u = \false$, $X_i = \start$, and $\EOS(\ell_i, v_i, \pi) = \true$}{
            $ST_u \gets \true$
            
            \IIf{$\ell_i$ is frozen}{\Return \false.}
            
            \IIf{$ST_w = \true$}{\Return \false.} 
        }
        \IIf{$X_i = \extend{}$ and $\EOE(l_i, v_i, \pi, ST_u) = \true$}{\Return \false.}
    }

\Return \true.
\end{algenv}
\end{figure}

\subsection{Correctness of the Oracles}

In this section, we prove the correctness of the vertex oracle. Namely, we prove that:

\begin{claim}\label{cl: vertex-oracle-correctness}
Let $v \in V$ and $ST, EX$ be the outputs of $\VO(v, \pi)$. It holds:
\begin{itemize}
    \item $ST = \true$ iff $v$ has an incident \start{} element in $\MS(G, \pi)$.
    \item $EX = \true$ iff $v$ has an incident \extend{} element in $\MS(G, \pi)$.
\end{itemize}
\end{claim}

Let us first prove two auxiliary claims about the correctness of the two edge oracles.

\begin{claim}\label{cl: start-oracle-correctness}
For any $\ell = ((u, w), \start{})$, if $\EOS(\ell, u, \pi)$ is called during computing $\VO(v, \pi)$, then $\EOS(\ell, u, \pi) = \true$ iff $\ell \in \MS(G, \pi)$.
\end{claim}
\begin{myproof}
We prove \cref{cl: start-oracle-correctness} by induction on $\pi(\ell)$. Suppose that the statement holds for all \start{} elements with a ranking lower than $\pi(\ell)$. If $\VO(v, \pi)$ directly calls $\EOS(\ell, u, \pi)$, it had already called $\EOS(\ell'',u,\pi)$ for every \start{} elements $\ell''$ incident on $u$ with a lower ranking. Moreover, if $\EOS(\ell, u, \pi)$ is called by $\EOS(\ell', w, \pi)$ or $\EOE(\ell', w, \pi, \cdot)$, then all \start{} elements on edges $(w, u')$ with a  smaller rank must be queried before $\ell$ by the description of oracles. All these calls to the edge oracle, \EOS{}, must return \false{} since the edge oracle queries $\EOS(\ell, u, \pi)$. By the induction hypothesis, there is no \start{} element incident to $w$ with a smaller rank in $\MS(G, \pi)$. Also, $\EOS(\ell, u, \pi)$ queries all the \start{} elements incident to $u$ with lower rank and returns \true{} if none of them is in $\MS(G, \pi)$. By the induction hypothesis, these calls are answered correctly, completing the proof.
\end{myproof}

\begin{claim}\label{cl: extend-oracle-correctness}
Let $\ell = ((u, w), \extend{})$, and let $ST_w$ indicate if $w$ has a \start{} element incident to it in $\MS(G, \pi)$ with a rank smaller than $\pi(\ell)$. If $\EOE(\ell, u, \pi, ST_w)$ is called during computing $\VO(v, \pi)$, then $\EOE(\ell, u, \pi, ST_w) = \true$ iff $\ell \in \MS(G, \pi)$.
\end{claim}
\begin{myproof}
We prove the statement by induction on $\pi(\ell)$. With a similar argument as in the proof of \Cref{cl: start-oracle-correctness}, we can show that \extend{} elements incident to $w$ with lower ranks are queried before $\ell$, and that the results of all these calls are \false{}. Also, if there exists a \start{} element incident to $w$ with a lower rank in $\MS(G, \pi)$, then by the description of \EOE, it is queried before and $ST_w$ is \true{}. Note that if such an edge exists, it must not be frozen, otherwise, the oracle does not query $\EOE(\ell, u, \pi, ST_w)$. Hence, if $\ell \notin \MS(G, \pi)$, there must be a frozen or an \extend{} element incident to $u$, or $ST_w = \true{}$ and there exists a \start{} element incident to $u$. Since $\EOE(\ell, u, \pi, ST_w)$ queries all the edges incident to $u$ with lower rank to determines if a \start{} or \extend{} element exists in $\MS(G, \pi)$, the proof is complete by induction hypothesis.
\end{myproof}

We are now ready to prove \cref{cl: vertex-oracle-correctness}.

\begin{proof}[Proof of \cref{cl: vertex-oracle-correctness}]
Since the vertex oracle queries the edge oracle \EOS{} for all incident \start{} elements in increasing order of their ranking until it finds a \start{} element in $\MS(G, \pi)$, by \Cref{cl: start-oracle-correctness}, the first property holds. Since \Cref{alg: vertex_oracle} and \Cref{alg: extend_oracle} query incident edges in increasing order, if $\EOE(\ell, u, \pi, ST_w)$ is queried for an \extend{} element $\ell = ((u, w), \extend{})$, $ST_w$ is computed correctly before calling $\EOE(\ell, u, \pi, ST_w)$ which implies that the condition in \Cref{cl: extend-oracle-correctness} holds. Therefore, with the similar argument for \extend{} elements and correctness of edge oracle \EOE{} by \Cref{cl: extend-oracle-correctness}, the second property holds.
\end{proof}

\subsection{Query Complexity of the Oracles}

In this section, we prove \cref{thm: query-complexity}. Consider $\ell = (e, X)$, an element in $T$. If $X = \start{}$, we write $Q(\ell, v, \pi)$ to denote the total number of recursive calls to $\EOS(\ell, \cdot, \pi, \cdot)$ during the execution of $\VO(v, \pi)$. If $X = \extend{}$, we write $Q(\ell, v, \pi)$ to denote the total number of recursive calls to $\EOE(\ell, \cdot, \pi, \cdot)$ during the execution of $\VO(v, \pi)$.

\begin{observation}\label{obs: one-start}
For a permutation $\pi$, at most one \start{} copy of every edge is queried in recursive calls of \EOS.
\end{observation}
\begin{myproof}
For a fixed edge, the edge oracle on its \start{} copies only generates a new recursive call on the first call because of the caching in \cref{line: caching} of \Cref{alg: start_oracle}.
\end{myproof}

\begin{observation}\label{obs: q-upperbound}
For every element $\ell = (e, X)$ in $T$ and permutation $\pi$, $Q(\ell, \pi) \leq O(n^2)$.
\end{observation}
\begin{myproof}
Let $e = \{x,y\}$. For a fixed vertex $u$, either the vertex oracle $\VO(u, \pi)$ queries the edge oracle for $\ell$ directly, or through some incident elemenet $\ell'$. Note that by \Cref{obs: one-start}, for each edge $e'$ incident to $e$, at most one of its $\start{}$ copies generates a new recursive call. Hence, the edge oracle of $\ell$ is called through at most $2(\deg(x) - 1) + 2(\deg(y) - 1)$ of its incident edges (one of its \start{} copy and one \extend{} element) which implies that there are at most $4(n - 1) + 1$ calls during the course of $\VO(u, \pi)$. In other words, we have that $Q(\ell, u, \pi) \leq 4n - 3$. Therefore,
$$Q(\ell, \pi) \leq \sum_{u\in V} Q(\ell, u, \pi) \leq n(4n - 3) \leq O(n^2).\qedhere$$
\end{myproof}

The main contribution of this section is to show that the expected $Q(\ell, \pi)$ for a random permutation $\pi$ is $O(\log^4 n)$ which is formalized in the following lemma.

\begin{lemma}\label{lem: query-complexity}
For any element $\ell \in T$, we have $\E_\pi[Q(\ell, \pi)] = O(\log^4 n)$.
\end{lemma}

Assuming the correctness of \Cref{lem: query-complexity}, we can complete the proof of \Cref{thm: query-complexity}.

\begin{proof}[Proof of \Cref{thm: query-complexity}]
\begin{align*}
    \E_{v, \pi}[F(v, \pi)] = \frac{1}{n}\E_\pi\left[\sum_{v\in V}F(v,\pi)\right] & = \frac{1}{n}\E_\pi\left[\sum_{v\in V}\sum_{\ell \in T}Q(\ell,v,\pi)\right] \\
    & = \frac{1}{n}\E_\pi\left[\sum_{\ell \in T}\sum_{v\in V}Q(\ell,v,\pi)\right] = \frac{1}{n} \E_\pi\left[\sum_{\ell \in T} Q(\ell, \pi)\right] \\ 
    & = \frac{1}{n} \sum_{\ell \in T} \E_\pi[ Q(\ell, \pi)] = \frac{1}{n} \sum_{\ell \in T} O(\log^4 n) \\
    & = \frac{1}{n} O(Km\cdot\log^4 n) = O(K\bar{d}\cdot \log^4 n).\qedhere
\end{align*}
\end{proof}

During the recursive calls to the edge oracles that start from $\VO(v, \pi)$, edges corresponding to the elements in the stack of recursive calls create a path. Because, we query $\VO(v, \pi)$ at the beginning, one of the endpoints of this path is vertex $v$. We direct the  elements in the path away from $v$ towards the other endpoint. We call such a directed path  starting from the bottom of the stack all the way to the top, a $(v,\pi)$-query-path. For an edge $e = (x,y)$, let $\vec{e}$ denote the directed edge from $x$ to $y$ and $\cev{e}$ denote a directed edge from $y$ to $x$. Similarly, for an element $\ell = ((x, y), X)$, $\vec{\ell}$ and $\cev{\ell}$ represent the direction of $e$.

Let $Q(\vec{\ell}, \pi) \subseteq Q(\ell, \pi)$ be the set of all query paths that end at $\vec{\ell}$ (with the same direction). In the following lemma, we bound the query complexity based on the direction of $\ell$. We use this lemma to prove \Cref{lem: query-complexity}. 

\begin{lemma}\label{lem: query-complexity-directed}
For any element $\ell \in T$, we have $\E_\pi[Q(\vec{\ell}, \pi)] = O(\log^4 n)$.
\end{lemma}

\begin{proof}[Proof of \Cref{lem: query-complexity}]
Since $Q(\ell, \pi) = Q(\vec{\ell}, \pi) \cup Q(\cev{\ell}, \pi)$, 
by \Cref{lem: query-complexity-directed} we have 
$$
\E_\pi[Q(\ell, \pi)] \leq \E_\pi[Q(\vec{\ell}, \pi)] + \E_\pi[Q(\cev{\ell}, \pi)] = O(\log^4 n) + O(\log^4 n) = O(\log^4 n). \qedhere
$$
\end{proof}

In the rest of this section, we focus on proving \Cref{lem: query-complexity-directed}. The first helpful observation is that all the \extend{} elements in a $(v,\pi)$-query-path appear before the \start{} elements.

\begin{observation}\label{obs: extend-then-start}
Let $\vec{P} = (\vec{\ell_r}, \ldots, \vec{\ell_1})$ be a $(v,\pi)$-query-path and $\vec{\ell_i} = (\vec{e_i}, X_i)$ for all $i$. Then there exists a $j$ $(0 \leq j \leq r)$ such that $X_i = \start{}$ for $0 < i \leq j$, and $X_i = \extend{}$ for $j < i \leq r$.
\end{observation}
\begin{proof}
Let $j$ be the maximum index such that $X_j = \start{}$. If there is no such index, the proof is complete since we can assume that $j = 0$. Now we claim that for all $i < j$, we have $X_i = \start{}$. This can be verified by noting that \cref{alg: start_oracle} only queries the edge oracle for \start{} elements. 
\end{proof}

Given a permutation $\pi$ and a path of elements $\vec{P}=(\vec{\ell_r}, \ldots, \vec{\ell_1})$, we define $\phi(\pi, \vec{P})$ to be another permutation $\sigma$ over edges such that:
\begin{align*}
    (\sigma(\ell_1), \sigma(\ell_2), \ldots, \sigma(&\ell_{r-1}), \sigma(\ell_r)) \coloneqq (\pi(\ell_2), \pi(\ell_3), \ldots, \pi(\ell_{r}), \pi(\ell_{1}))\\
    &\pi(\ell') = \sigma(\ell') \quad\quad \forall \ell' \notin \vec{P}.
\end{align*}

Given an element $\vec{\ell}$, by using the above mapping function we can construct a bipartite graph $H$ with two parts $A$ and $B$ such that each part has $((K+1)m)!$ vertices showing different permutations of elements. For a permutation $\pi \in A$ and a $(v,\pi)$-query-path $\vec{P}$ that ends at $\ell$ for some arbitrary vertex $v$, we connect $\pi$ in $A$ to $\phi(\pi, \vec{P})$ in $B$. Note that by the construction of $H$, $\deg(\pi_A) = Q(\vec{\ell},\pi_A)$ for all $\pi_A \in A$, since we have a unique element for each query-path that ends at $\vec{\ell}$ with permutation $\pi_A$. Hence, in order to prove \Cref{lem: query-complexity}, it is sufficient to prove that $\E_{\pi_A \sim A}[\text{deg}_H(\pi_A)]=O(\log^4 n)$. 

Let $\mathcal{Q}(\vec{\ell}, \pi)$ be the set of all query-paths for permutation $\pi$ that ends at $\vec{\ell}$. Let $\beta = c \log^2 n$ for some large constant $c$ and $\Pi$ be the set of all possible $((K+1)m)!$ permutations. We partition $\Pi$ into two subsets of {\em likely} and {\em unlikely} permutations, denoted by $L$ and $U$, respectively, as follows:
\begin{align*}
    L \coloneqq \left\{\pi \in \Pi \,\,\Big\lvert\,\, \max_{\vec{P} \in \mathcal{Q}(\vec{\ell}, \pi)} |\vec{P}| \leq \beta \right\} \qquad\qquad U \coloneqq \Pi \setminus L.
\end{align*}

In words, likely permutations are those where the longest query-path ending at $\vec{\ell}$ has a length of at most $\beta$ and unlikely permutations are the remaining ones. Let $A_L$ be the set of vertices corresponding to the likely permutations in $A$ and $A_U$ be the set of vertices corresponding to unlikely permutations. The intuition behind this partitioning is that the set of unlikely permutations makes up a tiny fraction of all permutations which is formalized in \Cref{lem: few-long-path}.

\begin{lemma}\label{lem: few-long-path}
If $c$ is a large enough constant, then we have $|A_U| \leq ((K + 1)m)! / n^2$.
\end{lemma}

Furthermore, we show that each vertex $\pi_B \in B$ has at most $O(\beta^2)$ neighbors among the likely permutations in part $A$ of our bipartite graph $H$.

\begin{lemma}\label{lem: few-likely-neighbors}
Let $\pi_B$ be a vertex in $B$. Then $\pi_B$ has at most $\beta^2$ neighbors in $A_L$. 
\end{lemma}

\begin{proof}[Proof of \Cref{lem: query-complexity-directed}]
We provide an upper bound on $|E(H)|$. First, note that by \Cref{lem: few-long-path}, we have $|A_U| \leq ((K + 1)m)! / n^2$. Furthermore, by \Cref{obs: q-upperbound}, degree of each vertex $\pi_A \in A$ is at most $O(n^2)$ which implies that the total number of edges incident to vertices of $A_U$ is at most 
$$E(A_U, B) \leq ((K + 1)m)! / n^2 \cdot O(n^2) \leq O\big(((K+1)m)!\big).$$

Moreover, by \Cref{lem: few-likely-neighbors}, each vertex $\pi_B \in B$ has at most $O(\beta^2)$ neighbors in $A_L$. Since $H$ is a bipartite graph, total number of incident edges to $A_L$ is at most $O(\beta^2)\cdot |A_L|$. Therefore, sum of degrees of all vertices in $A$ is at most 
\begin{align*}
    O(\beta^2)\cdot|A_L| + E(A_U, B) \leq O(\beta^2\cdot((K+1)m)!).
\end{align*}

For a random vertex in $A$, the expected degree is $\frac{O(\beta^2\cdot((K+1)m)!)}{((K+1)m)!} = O(\beta^2)$, which completes the proof since $\beta = c\log^2 n$ and $\deg(\pi_A) = Q(\vec{\ell},\pi_A)$.
\end{proof}

In the following two subsections, we prove \Cref{lem: few-likely-neighbors,lem: few-long-path}.

\subsubsection{Proof of \cref{lem: few-likely-neighbors}}
This proof is the most technical part of this part of the analysis. We show that for two query-paths $\vec{P}$ and $\vec{P'}$ that end at $\vec{\ell}$, and two permutations $\pi$ and $\pi'$, if $\phi(\pi, \vec{P}) = \phi(\pi', \vec{P'})$, then either one of the query paths is a subpath of the other one, or they ``branch'' through a specific configuration of \start{} and \extend{} elements illustrated in \cref{fig: valid-branch}.

To formalize this, let us formally define what a  ``branch'' is:

\begin{definition}\label{def: branch}
Let $\vec{P} = (\vec{\ell}_{r_1}, \ldots, \vec{\ell}_1)$ and $\vec{P'}=(\vec{\ell}_{r_2}', \ldots, \vec{\ell}_1')$ be two query-paths. Let $j$ be an such that $\vec{\ell}_i = \vec{\ell}_i'$ for all $i \leq j < \min(r_1, r_2)$, but $\vec{\ell}_{j+1} \neq \vec{\ell}_{j+1}'$. Then $\vec{P}$ and $\vec{P'}$ {\em branch} at element $\vec{\ell}_{j}$.
\end{definition}

And now let us define {\em valid} and {\em invalid} branches (see \cref{fig: valid-branch}).

\begin{definition}\label{def: valid-branch}
Let $\vec{P} = (\vec{\ell}_{r_1}, \ldots, \vec{\ell}_1)$ and $\vec{P'}=(\vec{\ell}_{r_2}', \ldots, \vec{\ell}_1')$ be two query-paths, where $\vec{\ell}_i = (\vec{e}_i, X_i)$, and $\vec{\ell}_i' = (\vec{e}_i', X_i')$. Assume that $\vec{P}$ and $\vec{P'}$ branch at element $\vec{\ell}_{j}$ for $j < \min(r_1, r_2)$. Moreover, let $u$ be the shared vertex between $\vec{e}_j$ and $\vec{e}_{j+1}$. We call this branch {\em valid} if $X_j = \start{}$ and $\{X_{j+1}, X_{j+1}'\} = \{\extend{}, \start{} \}$. For all other possible configurations of $X_j$, $X_{j+1}$, and $X_{j+1}'$, we say the branch is {\em invalid}.
\end{definition}

\begin{figure}[htbp]
\begin{center}
  \includegraphics[scale=0.8]{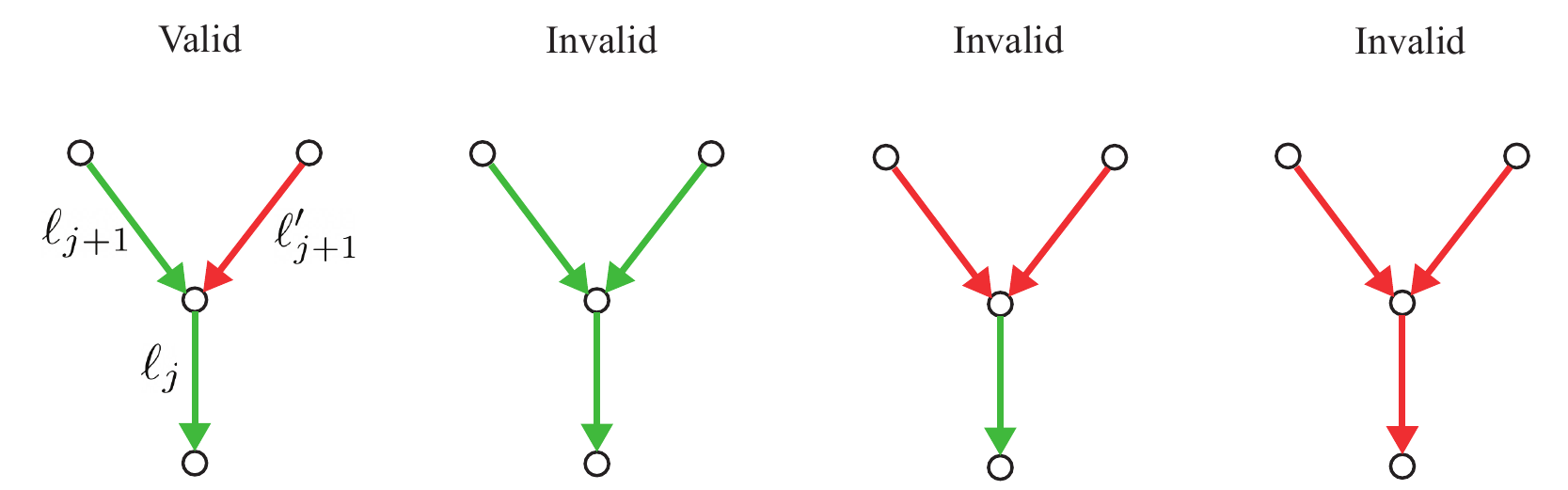}
  \caption{All valid and invalid branches. Green edges represent \start{} elements and red edges represent \extend{} elements.}\label{fig: valid-branch}
  \end{center}
\end{figure}

One key property of invalid branches is formalized in the following observation which is helpful in the rest of this section.

\begin{observation}
Let $\vec{P} = (\vec{\ell}_{r_1}, \ldots, \vec{\ell}_1)$ and $\vec{P'}=(\vec{\ell}_{r_2}', \ldots, \vec{\ell}_1')$ be two query-paths, $\vec{\ell}_i = (\vec{e}_i, X_i)$ and $\vec{\ell}_i' = (\vec{e'}_i, X_i')$ for all $i$. If there exists an invalid branch at $\vec{\ell}_{j}$, then none of the tuples $(X_j, X_{j+1})$, $(X_j, X_{j+1}')$, $(X_{j+1}, X_{j+1}')$, and $(X_{j+1}', X_{j+1})$ is $(\extend{}, \start{})$.
\end{observation}
\begin{proof}
Let $u$ be the shared endpoint between $\vec{e}_j$, $\vec{e}_{j+1}$, and $\vec{e'}_{j+1}$. If $X_j = \extend{}$, by \Cref{obs: extend-then-start}, we have $X_{j+1} = X_{j+1}' = \extend{}$. Therefore, all tuples are $(\extend{}, \extend{})$.

Now assume that $X_j = \start{}$. If one of $X_{j+1}$ or $X_{j+1}'$ is $\extend{}$ and the other one is $\start{}$, then the branch is valid by \Cref{def: valid-branch}. This leaves two remaining scenarios:
\begin{itemize}[leftmargin=10pt]
    \item $X_{j+1} = X_{j+1}' = \start{}$: In this case, all tuples are $(\start{}, \start{})$.
    \item $X_{j+1} = X_{j+1}' = \extend{}$: In this case, $(X_{j+1}, X_{j+1}') = (X_{j+1}', X_{j+1}) = (\extend{}, \extend{})$. Moreover, $(X_j, X_{j+1}) = (X_j, X_{j+1}') = (\start{},\extend{})$. Thus there is no $(\extend, \start)$ among the tuples considered in the statement.
\end{itemize}
The proof is thus complete.
\end{proof}

Next, we show that for two query-paths $\vec{P}$ and $\vec{P}'$ that end at $\vec{\ell}$, and two permutations $\pi$ and $\pi'$, if $\phi(\pi, \vec{P}) = \phi(\pi', \vec{P}')$, then either the shorter query-path is a subpath of the longer one or they have exactly one valid branch. In particular, it is not possible for $\vec{P}$ and $\vec{P}'$ to have an invalid branch or have multiple valid branches under this condition.

\begin{lemma}\label{lem:few-branch}
Let $\pi$ and $\pi'$ be two different permutations over $T$, and let $\vec{P}$ and $\vec{P}'$ be $(v,\pi)$- and $(v',\pi')$-query-path, respectively, that both end at element $\vec{\ell}$. If $\phi(\pi,\vec{P})=\phi(\pi',\vec{P}')$, then $\vec{P}$ and $\vec{P}'$ branch at most once and this branch is valid.
\end{lemma}

Before proving \Cref{lem:few-branch}, we prove a sequence of auxiliary observations and claims. Let $\vec{P} = (\vec{\ell}_{r_1}, \ldots, \vec{\ell}_1)$ and $\vec{P'}=(\vec{\ell}_{r_2}', \ldots, \vec{\ell}_1')$ where $\vec{\ell}=\vec{\ell}_1=\vec{\ell}_1'$. Also, let $\vec{\ell}_i=(\vec{e}_i, X_i)$ and $\vec{\ell'}_i = (\vec{e_i}', X_i')$. For the sake of contradiction, suppose that $\vec{P}$ and $\vec{P}'$ branch at element $\vec{\ell}_{b}$ and the branch is invalid. This means that $\vec{\ell}_i=\vec{\ell}_i'$ for $i \leq b$ and $\vec{\ell}_{b+1} \neq \vec{\ell}_{b+1}'$. Note that $\vec{e}_{b+1}$ and $\vec{e}_{b+1}'$ can be the same edge. We do not need to separately investigate these two cases as our proof generally works for both cases.

\begin{observation}\label{obs:increasing_path}
We have that $\pi(\ell_1) < \pi(\ell_2) < \ldots < \pi(\ell_{r_1})$ and $\pi'(\ell_1') < \pi'(\ell_2') < \ldots < \pi'(\ell_{r_2}')$.
\end{observation}
\begin{myproof}
\cref{alg: start_oracle,alg: extend_oracle} recursively call on elements with $\pi$ values less than $\pi$ value of the current element. Therefore, the stack of recursive calls will be decreasing with respect to $\pi$ values. The same condition also holds for permutation $\pi'$.
\end{myproof}

\begin{observation}\label{obs:equal_element}
$\pi(\ell_b) = \pi'(\ell_{b+1})$.
\end{observation}
\begin{proof}
Since $\vec{\ell}_{b+1}$ is not in $\vec{P}'$, we have that $\phi(\pi', \vec{P}')(\ell_{b+1})=\pi'(\ell_{b+1})$. Also, $\phi(\pi, \vec{P})(\ell_{b+1})=\pi(\ell_b)$ since $\phi(\pi, \vec{P})$ shifts the elements of path $\vec{P}$ by one. Given that $\phi(\pi, \vec{P})=\phi(\pi', \vec{P}')$, we get $\pi(\ell_b) = \pi'(\ell_{b+1})$.
\end{proof}

In the rest of the proof, we assume that $\pi(\ell_b) < \pi'(\ell_b)$. This is without loss of generality because we have not distinguished the two permutations $\pi$ and $\pi'$ in any other way.

\begin{observation}\label{obs:branch-begin}
$\pi'(\ell_{b+1}) < \pi'(\ell_b)$.
\end{observation}
\begin{myproof}
By combining \Cref{obs:equal_element}, our assumption that $\pi(\ell_b) < \pi'(\ell_b)$, and the fact that $\pi'$ is a permutation, we have that $\pi'(\ell_{b+1}) < \pi'(\ell_b)$.
\end{myproof}

\begin{claim}\label{clm:equal-ranks}
If $\pi(f) < \pi(\ell_b)$ or $\pi'(f) < \pi(\ell_b)$ for some element $f$, then $\pi(f) = \pi'(f)$.
\end{claim}
\begin{proof}
There are five different possible cases for $f$:
\begin{itemize}
    \item $f \notin P \cup P'$: Since $\phi$ only changes the rank of elements on the query-path and $\phi(\pi, \vec{P})(f)=\phi(\pi', \vec{P}')(f)$, we have $\pi(f) = \pi'(f)$.
    \item $f \in \{\ell_1, \ldots, \ell_{b-1}\}$: Since $\phi(\pi, \vec{P})(\ell_{i+1}) = \phi(\pi', \vec{P}')(\ell_{i+1})$ for $1 \leq i < b$, we have $\pi(\ell_i) = \pi'(\ell_i)$. Hence, $\pi(f) = \pi'(f)$.
    \item $f=\ell_b$: In this case, condition $\pi(f) < \pi(\ell_b)$ does not hold since $\pi(f) = \pi(\ell_b)$. Also, $\pi'(f)=\pi'(\ell_b) > \pi(\ell_b)$ by our assumption. Therefore, condition $\pi'(f) < \pi(\ell_b)$ does not hold.
    \item $f\in \{\ell_{b+1}, \ldots, \ell_{r_1}\}$: By \Cref{obs:increasing_path}, we have $\pi(f) > \pi(\ell_b)$. Therefore, condition $\pi(f) < \pi(\ell_b)$ does not hold. Let $f=\ell_i$ for $i > b$. Since $\phi(\pi, \vec{P}) = \phi(\pi', \vec{P}')$, we have that $\pi'(f) = \pi(\ell_{i-1}) \geq \pi(\ell_b)$. Therefore, none of the conditions in the claim statement hold.
    \item $f\in \{\ell_{b+1}', \ldots, \ell_{r_2}'\}$: By \Cref{obs:increasing_path}, we have $\pi'(f) > \pi'(\ell_b)$. Combined with our assumption $\pi'(\ell_b) > \pi(\ell_b)$, this gives $\pi'(f) > \pi(\ell_b)$. Let $f=\ell_i'$ for $i > b$. Since $\phi(\pi, \vec{P}) = \phi(\pi', \vec{P}')$, we have  $\pi(f) = \pi'(\ell_{i-1}) \geq \pi'(\ell_b) > \pi(\ell_b)$. Therefore, none of the conditions in the claim statement hold.
\end{itemize}

Each of the cases above either contradicts a condition of the claim, or satisfies $\pi(f) = \pi'(f)$. The proof is thus complete.
\end{proof}

\begin{observation}\label{obs: similar-freeze}
Let $f$ be a \start{} element such that $\pi(f) = \pi'(f)$. Then $f$ is frozen in permutation $\pi$ iff it is frozen in permutation $\pi'$. 
\end{observation}
\begin{myproof}
Since we fix the color of vertices, Bernoulli random variables in \cref{line:freeze} of \cref{alg: algorithm}, and the index $j^*$ that is used in \cref{alg: algorithm} for both permutations $\pi$ and $\pi'$ over $T$, the only way that one of the $f$ and $f'$ is frozen and the other one is not, is when $\pi(f) \neq \pi(f')$.
\end{myproof}

\begin{claim}\label{clm: impossible-branch}
$\ell_{b+1} \in \MS(G, \pi')$.
\end{claim}
\begin{proof}
We prove the claim by contradiction. Assume that $\ell_{b+1} \notin \MS(G, \pi')$. There are two possible scenarios for $\ell_{b+1}$ not to be in $\MS(G, \pi')$:
\begin{itemize}
    \item $X_{b+1} = \extend{}$ and $\ell_{b+1} \notin \MS(G, \pi')$ because of two \start{} elements $f, f' \in \MS(G, \pi')$:\\[0.3cm] Note that $\pi'(f) < \pi'(\ell_{b+1})$ and $\pi'(f') < \pi'(\ell_{b+1})$ since $f, f' \in \MS(G, \pi')$ and $\ell_{b+1} \notin \MS(G, \pi')$. On the other hand, by \Cref{obs:equal_element}, we have that $\pi(\ell_b) = \pi'(\ell_{b+1})$. Thus, $\pi'(f) < \pi(\ell_b)$ and $\pi'(f') < \pi(\ell_b)$. Hence, by \Cref{clm:equal-ranks}, $\pi(f) = \pi'(f)$ and $\pi(f') = \pi'(f')$, which implies that $f, f' \in \MS(G,\pi)$ since $\pi$ and $\pi'$ are identical for ranks smaller than $\pi(\ell_b)$. Moreover, by \Cref{obs: similar-freeze}, each of $f$ and $f'$ are either frozen in both $\pi$ and $\pi'$ or not. Also, both $f$ and $f'$ are not in path $P$ since $\pi(f) < \pi(\ell_b)$ and $\pi(f') < \pi(\ell_b)$. Let $e_{b+1} = (w, y)$ where $y$ is the shared endpoint with $e_{b}$. Without loss of generality, assume that $f$ is incident to $w$ and $f'$ is incident to $y$. Note that, both of $f$ and $f'$ cannot be incident to the same endpoint of $e_{b+1}$ since \start{} elements create a maximal matching. Since both edge oracle and vertex oracle queries edges in increasing order, when $\EOE(\ell_{b+1}, y, \pi, ST_w)$ was called, $ST_w$ must be \true{} since the edge oracle already queried element $f$ before. Furthermore, $\EOE(\ell_{b+1}, y, \pi, ST_w)$ calls edge oracle for $f'$ before $\ell_b$ since $\pi'(f) < \pi(\ell_b)$. This implies that $(v, \pi)$-query-path $\vec{P}$ is not a valid $(v, \pi)$-query-path since the $\EOE(\ell_{b+1}, y, \pi, ST_w)$ terminates after calling the edge oracle for element $f'$ (see \Cref{fig: invalid-path}).

    \item $\ell_{b+1} \notin \MS(G, \pi')$ because of a single element $f \in \MS(G, \pi')$:\\[0.3cm] Let $y$ be the shared endpoint of $e_b$ and $e_{b+1}$. With the same argument as in the previous case, we get that $f \in \MS(G, \pi)$ and $\pi(f) < \pi(\ell_b)$. Thus, by \Cref{obs: similar-freeze}, element $f$ is either frozen in both $\pi$ and $\pi'$ or in neither one. Note that either both of $f$ and $\ell_{b+1}$ are \start{} elements, or $f$ is a frozen element, or both of $f$ and $\ell_{b+1}$ are \extend{} elements. Hence, by \Cref{obs: extend-then-start}, $f$ must be queried before $\ell_{b+1}$ if it is not incident to $y$ in the edge oracle for permutation $\pi$ which implies that the edge oracle will not query $\ell_{b+1}$. Furthermore, if $f$ is incident to $y$, edge oracle queries $f$ before $\ell_b$ which implies that it terminates and $(v, \pi)$-query-path $\vec{P}$ is not a valid $(v, \pi)$-query-path. \qedhere
\end{itemize}
\end{proof}

\begin{figure}[htbp]
\begin{center}
  \includegraphics[scale=0.65]{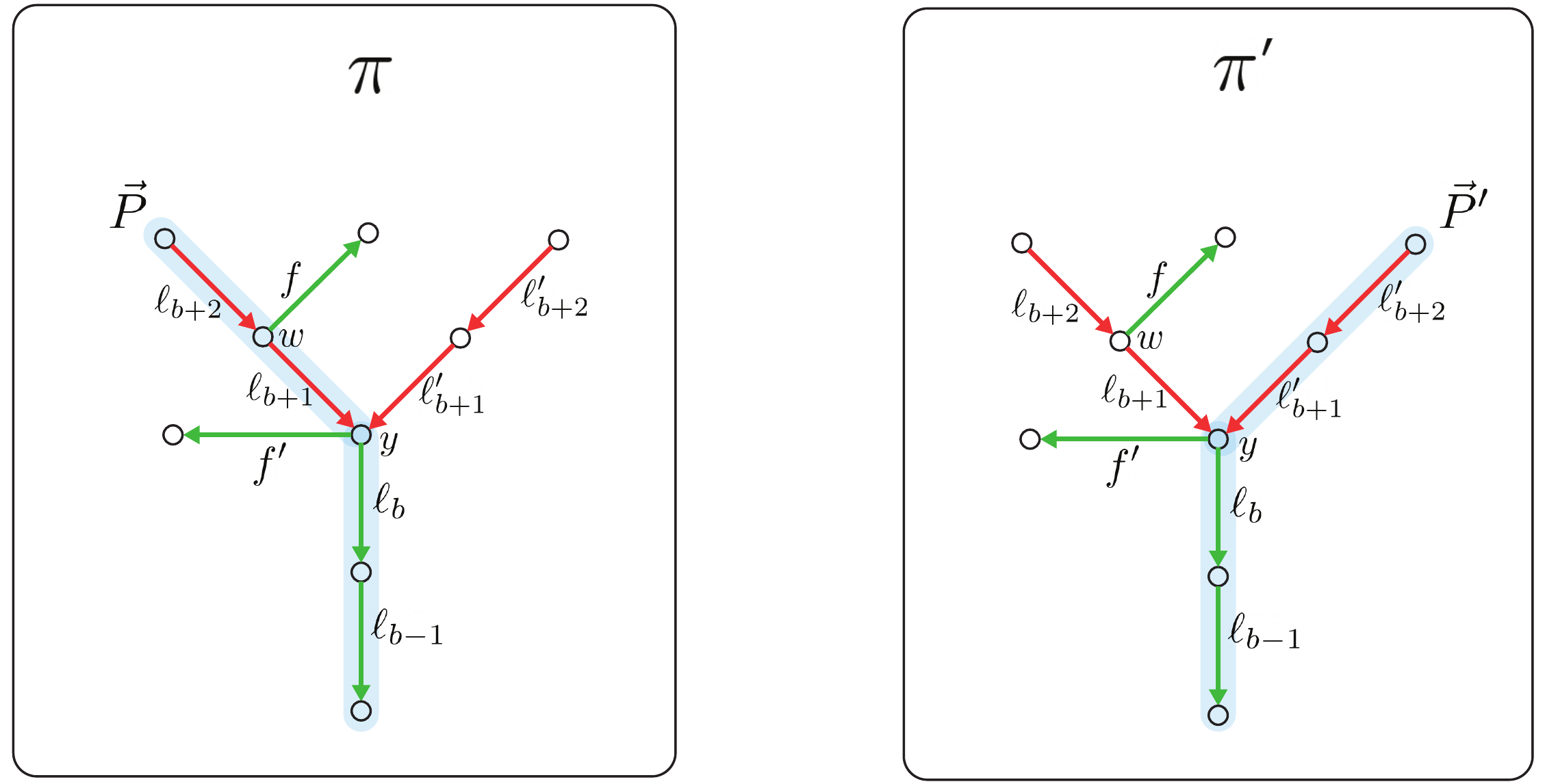}
  \caption{Illustration of a possible case of the first scenario in the proof of \Cref{clm: impossible-branch}. Green edges represent \start{} elements and red edges represent \extend{} elements. The left figure shows a query-path $\vec{P}$ in permutation $\pi$ which is highlighted in light blue. Similarly, the right figure shows a query-path $\vec{P}'$ in permutation $\pi'$. Query-path $\vec{P}$ is not a valid query-path since the $\EOE(\ell_{b+1}, y, \pi, ST_w)$ terminates after calling the edge oracle on element $f'$.}\label{fig: invalid-path}
  \end{center}
\end{figure}

\begin{proof}[Proof of \Cref{lem:few-branch}]
Assume for the sake of contradiction that the branch at $\vec{\ell}_{b}$ is an invalid branch. We prove that query-path $\vec{P}'$ is not a valid $(v, \pi')$-query-path. By \Cref{def: valid-branch}, we have that $X_{b+1} = X_{b+1}'$. Also, by \Cref{clm: impossible-branch}, edge oracle of $\ell_{b+1}'$ for permutation $\pi'$ queries $\ell_{b+1}$ before $\ell_b$ since $\pi'(\ell_{b+1}) < \pi'(\ell_b)$ by \Cref{obs:branch-begin}. Thus, the edge oracle for $\ell_{b+1}'$ terminates at this point. Therefore, $P$ and $P'$ contains no invalid branch.
\end{proof}

Now we are ready to complete the proof of \Cref{lem: few-likely-neighbors}.

\begin{proof}[Proof of \Cref{lem: few-likely-neighbors}]
As above, consider two query paths $\vec{P} = (\vec{\ell}_{r_1}, \ldots, \vec{\ell}_1)$ and $\vec{P'}=(\vec{\ell}_{r_2}', \ldots, \vec{\ell}_1')$ that end at $\vec{\ell} = \vec{\ell}_1 = \vec{\ell}_1'$ and suppose that $\phi(\pi,\vec{P})=\phi(\pi',\vec{P}')$. 

If $\ell$ is an \extend{} element, then all the elements in the two query paths $\vec{P}$ and $\vec{P}'$ must be \extend{} by \cref{obs: extend-then-start}. Therefore, $\vec{P}$ and $\vec{P'}$ cannot have a valid branch since a valid branch, by \cref{def: valid-branch}, requires two \start{} elements. Since \cref{lem:few-branch} asserts $\vec{P}$ and $\vec{P'}$ cannot have invalid branches either, one of $P$ and $P'$ must be a subpath of the other. Note that all paths that end at $\ell$ must be a subpath of the longest query path since we do not have a branch which implies that all paths must have different lengths. Therefore, there are at most $\beta$ query-paths that end at $\vec{\ell}$ since the length of the longest path is at most $\beta$.

Now suppose that $\ell$ is a \start{} element. Assume that $\vec{P}$ and $\vec{P}'$ branch at $\vec{\ell}_b$. Since this cannot be an invalid branch by \cref{lem:few-branch}, without a loss of generality, assume that $\ell_{b+1}'$ is an \extend{} element. By \Cref{obs: extend-then-start}, all $\ell_{b+1}', \ell_{b+2}', \ldots, \ell_{r_2}'$ must be \extend{} elements. Let $\mathcal{P}$ be the set of all query-paths that end at $\vec{\ell}$, and $\vec{P}_s \in \mathcal{P}$ be the path with maximum number of \start{} elements (break the tie with the longest length of \extend{} elements of the path). Note that there is no other \start{} element in the paths of $\mathcal{P}$ other than \start{} elements of $P_s$. Otherwise, we have an invalid branch since \start{} elements appear before \extend{} elements by \Cref{obs: extend-then-start}. 

Let $(\vec{\ell}_s, \ldots, \vec{\ell}_1)$ be the subpath of $\vec{P}_s$ consisting of \start{} elements. We claim that for each $\vec{\ell}_i$ ($i < s$), if there exist two query-paths in $\mathcal{P}$ that branch with $\vec{P}_s$ at $\vec{\ell}_i$, then one of them must be a subpath of the other. To see this, assume that there are two paths $\vec{P}$ and $\vec{P}'$ such that they both have a branch with $\vec{P}_s$ at $\vec{\ell}_i$. Let $\vec{f}$ and $\vec{f}'$ be the next elements on $\vec{P}$ and $\vec{P}'$. Note that it is possible that $\vec{f} = \vec{f}'$. By \Cref{def: valid-branch}, $f$ and $f'$ are \extend{} elements. Hence, the first elements that are not shared in both $\vec{P}$ and $\vec{P}'$ are \extend{} elements which means $\vec{P}$ and $\vec{P}'$ has an invalid branch.

\begin{figure}[htbp]
\begin{center}
  \includegraphics[scale=0.75]{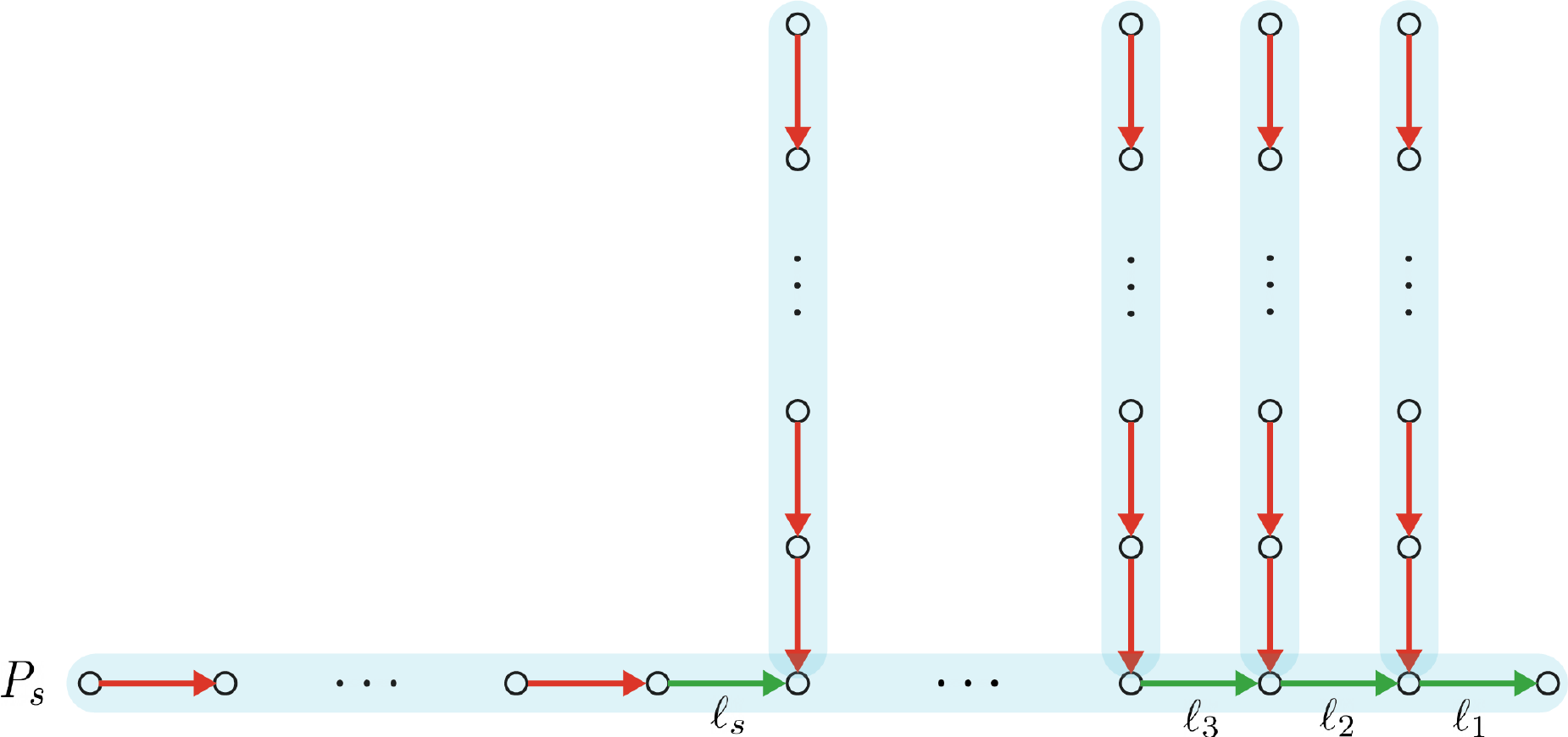}
  \caption{Illustration of $\hat{\mathcal{P}}$ and $P_s$. The green edges represent \start{} elements and the red edges represent \extend{} elements. Paths in $\hat{\mathcal{P}}$ are highlighted in light blue.}\label{fig: all-query-paths}
  \end{center}
\end{figure}

On the other hand, note that there is no path in $\mathcal{P}$ that has a branch with $\vec{P}_s$ in \extend{} elements of $\vec{P}_s$ since there is no invalid branch by \Cref{lem:few-branch}. Let $\hat{\mathcal{P}}$ be the union of $\vec{P}_s$ and the set of paths in $\mathcal{P}$ such that they have a valid branch with $\vec{P}_s$, and they are not a subgraph of other paths in $\mathcal{P}$ (see \Cref{fig: all-query-paths}).

By the above claim, $|\hat{\mathcal{P}}| \leq s + 1$. To complete the proof, assume that for each edge between $\pi_A \in A_L$ and $\pi_B \in B$ in graph $H$, we write a label $(\vec{P}', |\vec{P}|)$ where $\vec{P}$ is the query-path corresponding to the edge between $\pi_B$ and $\pi_A$, and $\vec{P}' \in \hat{\mathcal{P}}$ such that $\vec{P} \subseteq \vec{P}'$ (if there are multiple choices for $\vec{P}'$, choose the longest path). Since all labels must be different, and by definition of $A_L$ we have that all query-paths have a length of at most $\beta$, there are at most $(s+1)\beta \leq 2\beta^2$ different labels, where the last inequality followed by the fact that $|\vec{P}_s| \leq \beta$. 
\end{proof}

\subsubsection{Proof of \cref{lem: few-long-path}}
In this section, we prove that it is very unlikely to have long query-paths during the recursive calls of edge oracle. Our proof is inspired by \cite{Blelloch12}, who proved that the parallel round complexity of greedy maximal independent set is $O(\log^2 n)$. Our arguments are slightly different because our algorithm is not an instance of greedy MIS.

We define a $\theta$-prefix of permutation $\pi$ over elements $T$ to be the first $\theta |T|$ elements in permutation $\pi$. Suppose that instead of running the edge oracle on the whole permutation of elements, we choose a $\theta$-prefix of elements and run the edge oracle for elements of the prefix. The first useful observation is that if the algorithm calls edge oracle for one of the elements in a $\theta$-prefix of $T$, all recursive calls during this call will be on elements in the prefix.

\begin{observation}\label{obs: partition-query-lower-rank1}
Let $\ell$ be an element in $T$ and $\pi$ be a permutation over $T$. Then $\EOS(\ell, \cdot, \pi)$ or $\EOE(\ell, \cdot, \pi, \cdot)$ only recursively calls edge oracle for elements with a lower rank.
\end{observation}
\begin{proof}
Proof can be easily obtained by how we defined \Cref{alg: start_oracle} and \Cref{alg: extend_oracle} since each element only recursively calls elements with a lower rank.
\end{proof}

The high-level approach for this section is that we partition elements of $T$ to $O(\log n)$ continuous ranges and show that the length of the longest query-path inside each of these ranges is $O(\log n)$. Therefore, since for each element the edge oracle only calls elements with a lower rank, the length of the longest path in total should be at most $O(\log^2 n)$.

\noindent\textbf{Element Partitioning:} Let $\mathcal{P}_1$ be a $\theta_1$-prefix of $T$. Suppose that we run \Cref{alg: algorithm} on $\mathcal{P}_1$. Let $\mathcal{A}_1$ be the set of elements that are chosen by \Cref{alg: algorithm} for the selected subgraph $\MS(G, \pi)$. Also, let $\mathcal{D}_1$ be the set of elements outside the prefix $\mathcal{P}_1$ that cannot be in $\MS(G, \pi)$ due to the existence of some elements in $\mathcal{A}_1$. We delete $\mathcal{P}_1 \cup \mathcal{D}_1$ from $T$ and similarly choose $\theta_2$-prefix of the remaining elements. Let $\mathcal{P}_1, \mathcal{P}_2, \ldots, \mathcal{P}_d$ be the partitions with parameters $\theta_1, \theta_2, \ldots, \theta_d$, and $\mathcal{D}_1, \mathcal{D}_2, \ldots, \mathcal{D}_d$ be sets of deleted elements as defined above. We let $\theta_i = \Omega(2^i\log(n)/(2K\Delta))$ for the $i$-th partition and $d = O(\log n)$.

\begin{observation}\label{obs: partition-query-lower-rank2}
Let $\ell$ be an element in $\mathcal{P}_j \cup \mathcal{D}_j$ as defined above. Then $\EOS(\ell, \cdot, \pi)$ or $\EOE(\ell, \cdot, \pi, \cdot)$ only recursively calls edge oracle for elements in $\{\mathcal{P}_i\}_{i \leq j} \cup \{\mathcal{D}_i\}_{i < j}$.
\end{observation}
\begin{proof}
If $\ell \in \mathcal{P}_j$, since all the elements with lower ranks are in $\{\mathcal{P}_i\}_{i \leq j} \cup \{\mathcal{D}_i\}_{i < j}$, by \Cref{obs: partition-query-lower-rank1}, the proof is complete. If $\ell \in \mathcal{D}_j$, there are some elements in $\mathcal{P}_j$ that their existence in $\MS(G, \pi)$ caused $\ell$ not to be in $\MS(G, \pi)$. Hence, since edge oracle recursively calls incident elements in their increasing ranks, edge oracles for $\ell$ will only query incident elements in $\{\mathcal{P}_i\}_{i \leq j} \cup \{\mathcal{D}_i\}_{i < j}$.
\end{proof}

Note that the \start{} elements that appear in matching $M$ of \Cref{alg: algorithm} form a randomized greedy maximal matching, and our query process for \start{} elements also coincides with the query process for greedy matchings. Therefore, we can use the following result of \cite{behnezhad2021} which itself builds on the bound on \cite{FischerTALG} as black-box.

\begin{lemma}[{\cite[Lemma~3.13]{behnezhad2021}}]\label{lem: longest-start-path}
Let $\pi$ be a permutation over elements of $T$. With high probability, the maximum length of a query-path consisting of \start{} elements is $O(\log n)$. 
\end{lemma}

In the rest of this section, we show that it is unlikely to have a query-path consisting of \extend{} elements with length larger than $O(\log^2 n)$. Since by \Cref{obs: extend-then-start} all \start{} elements appear after \extend{} elements in a query-path, this is enough to show that the length of a query-path is bounded by $O(\log^2 n)$ with high probability.

The following two lemmas are similar to \cite[Lemma~3.1 and Lemma~3.3]{Blelloch12}, however, both are adapted to our setting.

\begin{lemma}\label{lem: reduce-degree}
Suppose that we choose a $\theta$-prefix of a uniformly at random permutation $\pi$ of elements of $T$, where $\theta = a/b$ for positive numbers $a \leq b \leq |T|$ such that $a/b \geq 2/|T|$. Let $\MS_{\theta}(G, \pi)$ be the subgraph produced by \Cref{alg: algorithm} on this prefix. Assume that we remove all other elements outside the prefix that cannot join subgraph $\MS(G, \pi)$ based on the current elements in $\MS_{\theta}(G, \pi)$. All remaining \extend{} elements have at most $b$ incident elements with probability of at least $1 - \frac{2|T|^2}{e^{a/2}}$.
\end{lemma}
\begin{proof}
We say an element is \textit{live} if we can add it to subgraph $\MS_{\theta}(G, \pi)$ and   \textit{dead} otherwise. Assume that we sequentially pick $(a|T|/b)$ elements. If the chosen element is live, we add the element to $\MS_{\theta}(G, \pi)$, and mark all incident elements that cannot be in $\MS_{\theta}(G, \pi)$ after adding this element, as dead. If the chosen element is dead, we do nothing. This sequential algorithm is equivalent to choosing a prefix of a permutation and then processing the permutation.

Let $\ell$ be an \extend{} element that is live and has more than $b$ incident live elements after processing the prefix. We show that this event is unlikely. Since at the end of $(a|T|/b)$ rounds, $\ell$ has more than $b$ incident live elements, before all of $(a|T|/b)$ rounds, it has more than $b$ incident live elements. Hence, in round $i$ of the sequential process, with probability of at least $b/(|T| - i) > b/|T|$, one of the incident live elements of $\ell$ will be selected. Note that if more than one incident element of $\ell$ is added to $\MS_{\theta}(G, \pi)$, then $\ell \not\in \MS(G, \pi)$. (It is possible that $\ell$ is not in $\MS(G, \pi)$ because of one incident element, however, we provide a looser bound by only considering the existence of two incident elements.) Thus, the probability that at most one of its incident elements is selected is at most
\begin{align}\label{eq: failure-probability}
    \left(1 - \frac{b}{|T|}\right)^{\frac{a|T|}{b}} + \left(\frac{a|T|}{b} \right) \left(1 - \frac{b}{|T|}\right)^{\frac{a|T|}{b}-1} < \frac{2a|T|}{b} \left(1 - \frac{b}{|T|}\right)^{\frac{a|T|}{2b}} = \frac{2a|T|}{b} \left(1 - \frac{b}{|T|}\right)^{\frac{|T|}{b} \cdot \frac{a}{2}},
\end{align}
where the first term of left-hand side is the probability that none of the incident elements is chosen and the second term is an upper bound for the probability that exactly one of the incident elements is selected. Combining equations (\ref{eq: failure-probability}), $(1 - \frac{b}{|T|})^{\frac{|T|}{b}} < (1/e)$, and $\frac{a}{b} < 1$, the probability of this event is at most $\frac{2|T|}{e^{a/2}}$. Taking a union bound over all elements completes the proof. 
\end{proof}

\begin{corollary}
If $\theta_i = \Omega(2^i \log n / (2K\Delta))$, then all remaining elements have at most $2K\Delta / 2^i$ incident elements after round $i$. 
\end{corollary}
\begin{proof}
At the beginning, each element is incident to at most $2K\Delta$ elements since the degree of each vertex is at most $\Delta$ and we create $K$ copies of each edge. Since $\log |T| = O(\log n)$, the proof can be obtained by \Cref{lem: reduce-degree} with $a = \Omega(\log n)$ and $b = 2K\Delta/2^i$ for partition $i$.
\end{proof}

In the previous lemma, by choosing the appropriate $\theta_i$, the number of incident elements for each \extend{} element reduce by half after removing partition $i$. Hence, there will be at most $O(\log n)$ partitions. In the next lemma, we will show that the length of the longest query-path for \extend{} elements in each of $\mathcal{P}_i$ is bounded by $O(\log n)$.

\begin{lemma}\label{lem: longest-partition-path}
Suppose that all \extend{} elements have at most $x$ incident elements. Let $a$ and $b$ be two positive integers such that $a \geq b$ and consider a randomly ordered $\theta$-prefix from $T$ with $\theta < b/x$. Then the longest query-path consist of \extend{} elements has length $O(a)$ with probability of at least $1 - |T|(b/a)^a$.
\end{lemma}
\begin{proof}
Let $(i_1, i_2, \ldots, i_{k+1})$ be $k+1$ different indices in the $\theta$-prefix. Choosing the prefix elements sequentially is equivalent to choosing a randomly ordered prefix. Let $(\ell_{i_1}, \ell_{i_2}, \ldots, \ell_{i_{k+1}})$ be elements in these indices that create a query-path. Hence, the probability that $\ell_{i_1}$ and $\ell_{i_2}$ are incident is at most $x/(|T|-1)$ since when we select $\ell_{i_2}$, element $\ell_{i_1}$ is already chosen and there are $|T| - 1$ choices remaining and $\ell_1$ has at most $x$ incident elements. With the same argument, the probability that $\ell_2$ and $\ell_3$ be incident is at most $x / (|T| - 2)$. Hence, the probability of having such a path is at most $(x / (|T| - k))^k$. Taking a union bound over all possible $k+1$ indices of the prefix, we get the following bound for having a query-path of length $k+1$. By assuming that $k < |T|/2$, we have
\begin{align*}
    {\theta|T| \choose k+1}\cdot (x / (|T| - k))^k &\leq \left(\frac{e \theta |T|}{k+1}\right)^{k+1}\cdot \left(\frac{x}{|T| - k}\right)^k \\
    & = \frac{e\theta|T|}{k+1} \cdot \left(\frac{ex\theta|T|}{(k+1)(|T| - k)}\right)^k \\
    & \leq \frac{e\theta|T|}{k+1} \cdot \left(\frac{2ex\theta}{k+1}\right)^k & (k \leq |T|/2) \\
    & \leq |T|\left(\frac{2ex\theta}{k+1}\right)^{k+1}.
\end{align*}
By setting $k = 2ea - 1$ and $\theta < b/x$ the above bound will be at most $|T|(b/a)^a$. Therefore, with probability $1 - |T|(b/a)^a$, the longest query-path has length $O(a)$. Note that if $k \geq |T|/2$, it implies that $2ea \geq |T|/2$ which the lemma clearly holds since $a = O(|T|)$.
\end{proof}

\begin{corollary}
Suppose that all \extend{} elements have at most $x$ incident elements. The longest query-path consisting of \extend{} elements, has length of at most $O(\log n)$ with probability $1 - \frac{1}{n^4}$ for a $O(\log(n)/x)$-prefix of $T$. 
\end{corollary}
\begin{proof}
Setting $a=5b=O(\log |T|) = O(\log n)$ in \Cref{lem: longest-partition-path} completes the proof.
\end{proof}

Now we are ready to complete the proof of \Cref{lem: few-long-path}.

\begin{proof}[Proof of \Cref{lem: few-long-path}]
First, if the edge oracle recursively calls a \start{} element, then we get from \Cref{lem: longest-start-path} and \Cref{obs: extend-then-start} that the remaining length of the query-path will not exceed $O(\log n)$ with high probability. Therefore, we only need to show that the length of the longest query-path involving only \extend{} elements is bounded by $O(\log^2 n)$ with high probability. Note that according to the partitioning, if we choose $\theta_i = \Omega(2^i \log(n)/(2K\Delta))$, by \Cref{lem: reduce-degree}, after round $i$ each \extend{} element has at most $(2K\Delta)/2^i$ incident elements. This implies that the number of parts is $O(\log n)$ (i.e. $d = O(\log n)$). Furthermore, by \Cref{lem: longest-partition-path}, the longest path consisting of \extend{} elements in each part has length at most $O(\log n)$.

\begin{figure}[htbp]
\begin{center}
  \includegraphics[scale=0.55]{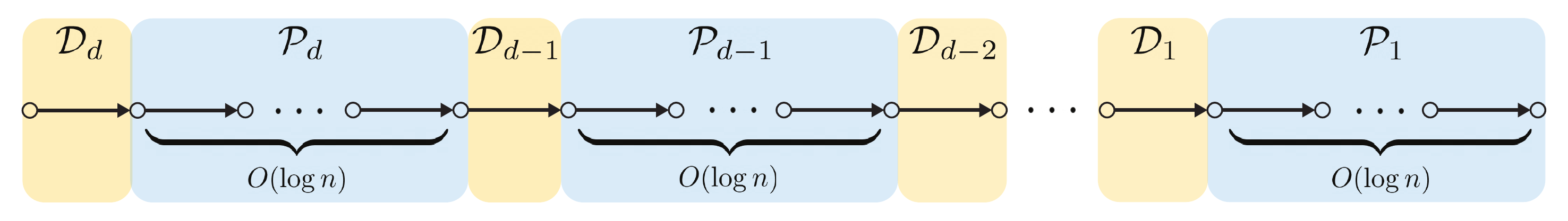}
  \caption{A possible query-path according to the partitioning. Blue boxes represent $\mathcal{P}_i$ and yellow boxes represent $\mathcal{D}_i$.}\label{fig:partitioning}
  \end{center}
\end{figure}

Now consider a query-path $P$. By \Cref{obs: partition-query-lower-rank2}, at most one element of $P$ is in each $D_i$ for $i \leq d$. Moreover, by the above argument, there are at most $O(\log n)$ elements of $P$ in each of $\mathcal{P}_i$ (see \Cref{fig:partitioning}). Therefore, the longest length of a query-path is bounded by $O(\log^2 n)$ with high probability (i.e. with probability of at least $1-1/n^2$).

For each unlikely permutation $\pi \in U$, there exists a query-path of length larger than $\beta$. Since $\beta = c\log^2 n$, by choosing $c$ large enough we have that $|U|/|\Pi| \leq 1/n^2$. Therefore, $|U| \leq ((K+1)m)!/n^2$ which implies that $|A_U| \leq ((K + 1)m)!/n^2$ since $A_U$ represents vertices that correspond to unlikely permutations.
\end{proof}

\section{Our Estimator for the Adjacency List Model}
In this section, we use the oracles introduced in the previous section to estimate the size of the matching in \Cref{alg: algorithm}. We assume that without loss of generality, the algorithm knows $\Delta$, $\bar{d}$, and there is no singleton vertex in the graph (note that the algorithm can simply query degree of each vertex and compute $\Delta$ and $\bar{d}$).

Note that our upper bound of \cref{sec: query-process} is on $F(v, \pi)$ which recall is the number of recursive calls to the oracles, but not necessarily the running time needed to implement it. Generating the whole permutation $\pi$ requires $|T| = \Theta(K m)$ time, which is too large for our purpose. Therefore, we have to generate $\pi$ on the fly during the recursive calls whenever needed. Using the techniques developed first by \cite{OnakSODA12} and further used by \cite{behnezhad2021}, we show in \cref{sec: implementation} that indeed it is possible to get an $\widetilde{O}(F(v, \pi))$ time implementation. In particular, we show that:

\begin{restatable}{lemma}{oracleImplementation}\label{lem: oracle-implementation}
In the adjacency list model, there is a data structure that given a graph $G$, (implicitly) fixes a random permutation $\pi$ over its edge set. Then for any vertex $v$, the data structure returns whether $v$ has any edge in outputs $M$ and $S$ of \cref{alg: algorithm} according to a random permutation $\pi$. Each query $v$ to the data structure is answered in $\Tilde{O}(F(v, \pi))$ time w.h.p. where $F(v, \pi)$ is as defined in \cref{sec: query-process}. Additionally, the vertices we feed into the oracle can be adaptively chosen depending on the responses to the previous calls.
\end{restatable}

In order to estimate the output of our algorithm, we sample $r$ random vertices, and for each vertex, we check if it is the endpoint of a \start{} element. Moreover, for each sampled vertex, we check if the vertex is an endpoint of a length three augmenting path that is created by a \start{} element as a middle edge and two other \extend{} elements.

\subsection{Multiplicative Approximation}

We use the following well-known claim about the size of maximum matching, a similar bound to which was also used in \cite{behnezhad2021}.

\begin{claim}\label{clm: matching-size}
Let $G$ be a graph with maximum degree $\Delta$ and average degree $\bar{d}$. Then $\mu(G) \geq \frac{n\bar{d}}{4\Delta}$.
\end{claim}
\begin{myproof}
    Greedily edge color the graph using $2\Delta$ colors. The color with the largest size is a matching of size at least $m/2\Delta = n \bar{d}/4\Delta$.
\end{myproof}

We use this claim to show that the number of samples that is needed to estimate the output of \Cref{alg: algorithm} is $r=\Tilde{\Theta}(\Delta/\bar{d})$.

\begin{figure}[h]
\begin{algenv}{Final algorithm for adjacency list.}{alg: adjacency-list}
    $r \gets (384
\cdot \Delta \log n) / (\delta^{2}\bar{d})$.
    
    Sample $r$ vertices $u_1, u_2, \ldots, u_r$ uniformly at random from $V$ with replacement.
    
    Run vertex oracle for each $u_i$ and let $S_i$ and $E_i$ be the indicator that $u_i$ has an incident \start{} and \extend{} elements that appear in the output of \cref{alg: algorithm}.
    
    \For{$i$ in $1 \ldots r$}{ \label{line: forloop}
    
    $X_i \gets 0$
    
    \lIf{$S_i = 1$}{$X_i \gets 1$}
    
    \If{$S_i = 0$ and $E_i = 1$}{
        Let $w_i$ be the other endpoint of \extend{} element incident to $u$.
        
        Run vertex oracle for $w_i$ (with the same permutation) and let $S'_i$ be the indicator that $w_i$ has an incident \start{} element.
        
        \lIf{$S'_i = 0$}{\Continue}
        
        Let $x_i$ be other endpoint of \start{} element incident to $w_i$.
        
        Run vertex oracle for $x_i$ (with the same permutation) and let $E'_i$ be the indicator that $x_i$ has an incident \extend{} element.
        
        \lIf{$E'_i = 0$}{\Continue}
        
        Let $y_i$ be other endpoint of \extend{} element incident to $x_i$.
        
        Run vertex oracle for $y_i$ (with the same permutation) and let $S''_i$ be the indicator that $y_i$ has an incident \start{} element.
        
        \lIf{$S''_i = 0$}{$X_i \gets 1$}
    }
    }

Let $X = \sum_{i \in [r]}X_i$ and $f = X/r$.

Let $\tilde{\mu} = (1 - \frac{\delta}{2})fn/2$.
    
\Return $\tilde{\mu}$

\end{algenv}
\end{figure}

\begin{remark}\label{rem: three-augmenting}
In \Cref{alg: adjacency-list}, we do not estimate $\mu(M \cup S)$. Instead, we estimate the size of the maximal matching $M$ and augment length-three augmenting paths in $M \cup S$. Since the bound in \Cref{thm:apx} is based on augmenting length three augmenting paths of $M \cup S$, we get the same approximation guarantee.
\end{remark}

We can change the vertex oracle and edge oracle to return the incident elements that appear in subgraph $M \cup S$. However, for simplicity of oracles, we return the indicators for having \start{} and \extend{} elements but we assume that we have access to these elements.

\begin{lemma}\label{lem: multi-error-bound-adj}
Let $\tilde{\mu}$ be the output of \cref{alg: adjacency-list}. With high probability,
$$\left(\frac{1}{2} + \frac{\delta}{4}\right)\mu(G) \leq \tilde{\mu} \leq \mu(G),$$
where $\delta$ is as in the statement of \cref{thm:apx}. 
\end{lemma}

\begin{proof}
Let $X_i$ be the indicator for each vertex $i$ which shows in output of \Cref{alg: algorithm}, either $i$ has a \start{} incident element, or it is an endpoint of a length three augmenting path created by a \start{} element in the middle along with two \extend{} elements. Note that the way $X_i$ is computed in \cref{alg: adjacency-list} is the same as the definition given earlier.

Let $\hat{M}(G, \pi)$ be the set of edges of the matching  created by augmenting the length-three augmenting paths of $M \cup S$. By \Cref{rem: three-augmenting} and \Cref{thm:apx},
\begin{align}\label{eq: alg-bound-3aug}
    \left(\frac{1}{2} + \delta\right)\mu(G) \leq \E_{\pi}|\hat{M}(G, \pi)| \leq \mu(G).
\end{align}
Since the number of matched vertices is twice the number of edges in the matching,  
\begin{align*}
    \E[X_i] = \Pr[X_i = 1] = \frac{2 \E_\pi|\hat{M}(G, \pi)|}{n}.
\end{align*}
Similarly, we have 
\begin{align}\label{eq: X-expected}
    \E[X] = \frac{2r \E_\pi|\hat{M}(G, \pi)|}{n}.
\end{align}
Since $X$ is the sum of $r$ independent Bernoulli random variables, the Chernoff bound (\Cref{prop:chernoff}) implies
\begin{align}\label{eq: with-high-prob}
    \Pr[|X - \E[X]| \leq \sqrt{6 \E[X] \log n}] \leq 2 \exp \left(-\frac{6 \E[X] \log n}{3 \E[X]} \right) = \frac{2}{n^2}.
\end{align}
Now, $fn = Xn / r$. Therefore, using the above equation, $fn$ is in the following range with probability $1 - 2/n^2$.
\begin{align*}
    fn \in \frac{n(\E[X] \pm \sqrt{6 \E[X] \log n)}}{r} & = \frac{n\E[X]}{r} \pm \frac{\sqrt{6 n^2\E[X] \log n}}{r}\\
    & = 2\E|\hat{M}(G, \pi)| \pm \sqrt{\frac{12n\E|\hat{M}(G, \pi)| \log n}{r}} & \tag{By (\ref{eq: X-expected})}\\
    & = 2\E|\hat{M}(G, \pi)| \pm \sqrt{\frac{n\E|\hat{M}(G, \pi)|\bar{d}}{32  \cdot \delta^{-2} \cdot \Delta}} \tag{Since $r=\frac{ 384
\cdot \Delta \log n}{\delta^2 \bar{d}}$} \\
    & \geq 2\E|\hat{M}(G, \pi)| \pm \sqrt{\frac{\mu(G)\E|\hat{M}(G, \pi)|}{8 \cdot \delta^{-2}}} & \tag{Since $\mu(G) \geq \frac{n\bar{d}}{4\Delta}$}
\end{align*}
By (\ref{eq: alg-bound-3aug}), we have $2\E|\hat{M}(G, \pi)| \geq \mu(G)$. Combining with \Cref{clm: matching-size}, we get
\begin{align*}
    fn \in 2\E|\hat{M}(G, \pi)| \pm \sqrt{\frac{\E|\hat{M}(G, \pi)|^2}{4\delta^{-2}} } = (2 \pm \frac{\delta}{2}) \E|\hat{M}(G, \pi)|.
\end{align*}
Since $\tilde{\mu} = (1 - \frac{\delta}{2})fn/2$, we have that
\begin{align*}
    (1-\delta)\E|\hat{M}(G, \pi)| \leq \tilde{\mu} \leq \E|\hat{M}(G, \pi)|
\end{align*}
Next, note that by (\ref{eq: alg-bound-3aug}), $(\frac{1}{2} + \delta)\mu(G) \leq \E|\hat{M}(G, \pi)| \leq \mu(G)$. Hence,
\begin{align*}
    (1-\delta)\left(\frac{1}{2} + \delta\right)\mu(G) \leq \tilde{\mu} \leq \mu(G),
\end{align*}
and thus
$$
    \left(\frac{1}{2 } + \frac{\delta}{4}\right)\mu(G) \leq \tilde{\mu} \leq \mu(G). \qedhere
$$
\end{proof}

\begin{lemma}\label{lem: adjcacency-list-time}
\Cref{alg: adjacency-list} runs in $\Tilde{O}(n + \Delta^{1+\epsilon})$ with high probability.
\end{lemma}
\begin{proof}
We show that for each of $r$ sampled vertices in \Cref{alg: adjacency-list}, the iteration corresponding to the vertex in Line \ref{line: forloop} takes $\Tilde{O}(K\bar{d})$ time. First, note that for a random permutation $\pi$, the vertex oracle for a vertex $u$ can be called a constant number of times. The reason is that the algorithm only queries vertex oracle for vertices of length three augmenting paths. Also, by \Cref{thm: query-complexity}, the vertex oracle takes $\Tilde{O}(K\bar{d})$ time for random vertex and permutation. Therefore, if we run the vertex oracle a constant time for a fixed vertex, still the expected running time will be $\Tilde{O}(K\bar{d})$. Since the number of samples is $O(\Delta \log n / \bar{d})$, the total running time for all samples will be $\Tilde{O}(K\Delta)$. Furthermore, we spent $O(n)$ time for computing $\Delta$, $\bar{d}$, and finding the isolated vertices. 

In order to achieve a high probability bound on running time, we run $\Theta(\log n)$ instances of the algorithm simultaneously and return the estimation of the first instance that terminates. Since the expected running time is $\Tilde{O}(n + K\Delta)$, the first instance terminates with probability $1 - 1/\text{poly}(n)$ in $\Tilde{O}(n + K\Delta)$.

Plugging $K = \tilde{O}(\Delta^{\epsilon})$ completes the proof.
\end{proof}

\subsection{Multiplicative-Additive Approximation}
We use the same algorithm as \Cref{alg: adjacency-list}, however, the $o(n)$ additive error allows us to sample $\Tilde{\Theta}(1)$ vertices instead of $\Tilde{\Theta}(\Delta/\bar{d})$ vertices. Moreover, we no longer need to estimate $\bar{d}$ and $\Delta$ since the number of samples is independent of these parameters.

\begin{lemma}\label{lem: multi-add-error-bound-adj}
Let $\tilde{\mu}$ be the output of \cref{alg: adjacency-list} with parameter $r = 12 \log^3 n$ and estimation $\Tilde{\mu} = fn/2 - \frac{n}{2 \log n}$. With high probability,
$$ \left(\frac{1}{2} + \delta \right)\mu(G) - \frac{n}{\log n} \leq \tilde{\mu} \leq \mu(G).$$
\end{lemma}
\begin{proof}
Let $X_i$ be defined the same as \Cref{lem: multi-error-bound-adj}. With the exact same argument, inequalities \Cref{eq: alg-bound-3aug}, \Cref{eq: X-expected}, and \Cref{eq: with-high-prob} hold with new parameter and estimation. Hence, with probability of at least $1 - 2/n^2$,

\begin{align*}
    fn \in \frac{n(\E[X] \pm \sqrt{6 \E[X] \log n)}}{r} & = \frac{n\E[X]}{r} \pm \sqrt{\frac{6 n^2\E[X] \log n}{r^2}}\\
    & = 2\E|\hat{M}(G, \pi)| \pm \sqrt{\frac{12n\E|\hat{M}(G, \pi)| \log n}{r}} & (\text{By }(\ref{eq: X-expected}))\\
    & = 2\E|\hat{M}(G, \pi)| \pm \sqrt{\frac{n\E|\hat{M}(G, \pi)|}{\log^2 n}} & (\text{Since } r = 12 \cdot \log^3 n) \\
    & \in 2\E|\hat{M}(G, \pi)| \pm \frac{n}{\log n} & (\text{Since } \E|\hat{M}(G, \pi)| \leq n).
\end{align*}
Since $\tilde{\mu} = fn/2 - \frac{n}{2 \log n}$, we have that
\begin{align*}
    \E|\hat{M}(G, \pi)| - \frac{n}{\log n} \leq \tilde{\mu} \leq \E|\hat{M}(G, \pi)|.
\end{align*}
By plugging \Cref{eq: alg-bound-3aug}, we get 
$$
\left(\frac{1}{2} + \delta \right)\mu(G) - \frac{ n}{\log n} \leq \tilde{\mu} \leq \mu(G). \qedhere
$$
\end{proof}

\begin{lemma}\label{lem: adjcacency-list-additive-time}
\Cref{alg: adjacency-list} with parameter $r = 12 \log^3 n$, runs in $\Tilde{O}(\bar{d} \cdot \Delta^\epsilon)$ with high probability.
\end{lemma}
\begin{proof}
First, note that we do not need to spend $\Tilde{O}(n)$ time to estimate $\Delta$ and $\bar{d}$, since the number of samples is independent of these parameters. By the exact same argument as the proof of \Cref{lem: adjcacency-list-time}, we can show that the expected running time for each sampled vertex is $\Tilde{O}(K\bar{d})$. Since $r = \Tilde(O)(1)$, the total running time will be $\Tilde{O}(K\bar{d})$.

To get a high probability bound on the running time, similar to \Cref{lem: adjcacency-list-time}, we run $\Theta(\log n)$ instances of the algorithm simultaneously. Using the same argument, the first instance terminates with probability $1 - 1/\text{poly}(n)$ in $\Tilde{O}(K\bar{d})$.

Plugging $K = \tilde{O}(\Delta^{\epsilon})$ completes the proof.
\end{proof}

\section{Our Estimator for the Adjacency Matrix Model}
In this section, we give a reduction from the adjacency matrix model to the adjacency list model such that each query in the adjacency list can be implemented with a constant number of queries in the adjacency matrix model. Such a reduction appears in \cite[Section~5]{behnezhad2021}. We use a similar idea with a minor modification in the parameters of the construction.

Let $\gamma = (4\log n)\cdot n$. We construct a graph $H = (V_H, E_H)$ as follows:

\begin{itemize}
    \item $V_H$ is the union of $V_1, V_2$ and $U_1, U_2, \ldots, U_n$ such that:
    \begin{itemize}
        \item $V_1$ and $V_2$ are two copies of the vertex set of the original graph $G$.
        \item $U_i$ is a vertex set of size $\gamma$ for each $i \in [n]$.
    \end{itemize}
    \item We define the edge set such that degree of each vertex is in $\{1, n, n + \gamma \}$:
    \begin{itemize}
        \item Degree of each vertex $v \in V_1$ is $n$. The $i$-th neighbor of $v$ is the $i$-th vertex in $V_1$ if $(v, i) \in E$, otherwise, its $i$-th neighbor is the $i$-th vertex in $V_2$ for $i \leq n$. Note that graph $(V_1, E_H \cap (V_1 \times V_1))$ is isomorphic to $G$.
        \item Degree of each vertex $v \in V_2$ is $n + \gamma$. The $i$-th neighbor of $v$ is the $i$-th vertex in $V_2$ if $(v, i) \in E$, otherwise, its $i$-th neighbor is the $i$-th vertex in $V_1$ for $i \leq n$. For all $n < i \leq n + \gamma$, the $i$-th neighbor of $v$ is $i$-th vertex in $U_v$.
        \item Degree of each vertex $u \in U_i$ is one for $i \in [n]$. The only neighbor of $u$ is the $i$-th vertex of $V_2$.
    \end{itemize}
\end{itemize}

\begin{observation}\label{obs: finding-neighbor}
For each vertex $v \in V_H$ and $i \in [\deg_H(v)]$, the $i$-th neighbor of vertex $v$ can be determined using at most one query to the adjacency matrix.
\end{observation}
\begin{proof}
First, note that the degree of each vertex is not dependent on the original graph $G$. Hence, we do not need any queries to find the degree of each vertex. For each vertex $v \in U_1 \cup U_2 \cup \ldots \cup U_n$, one can find its only neighbor without any queries by the construction of $H$. For each vertex $v \in V_1 \cup V_2$, the $i$-th neighbor is either the $i$-th vertex of $V_1$ or the $i$-vertex of $V_2$. Therefore, with at most one query, one can determine the $i$-th neighbor of vertex $v$.
\end{proof}

Consider a random permutation $\pi$ over the list of elements $T$, consisting of \start{} and \extend{} copies of $E_H$. Intuitively, for almost all vertices $v \in V_2$, the first incident \start{} element to $v$ in $T$ is an edge between $v$ and $U_v$. Similarly, the first two incident \extend{} elements to $v$ are between $v$ and $U_v$. We say $v$ is a \textit{bad} vertex, if it violates the mentioned conditions. Let $R \subseteq V_2$ be the set of bad vertices. Since the permutation over edges in $H[V_1 \cup R]$ is uniformly at random, using \Cref{thm:apx}, we can provide a bound on the size of the matching produced by the algorithm.

\begin{observation}\label{obs: good-vertices}
For each $v \in V_2 \setminus R$, the incident \start{} and \extend{} elements in $\MS(H, \pi)$ are between $v$ and $U_v$. Moreover, both incident \start{} and \extend{} elements in $\MS(H, \pi)$ appears before all edges between $v$ and $V_1 \cup V_2$.
\end{observation}
\begin{proof}
By the definition of $R$, the first \start{} element incident to $v$ in the permutation is between $v$ and $U_v$. Let this edge be $(v, w)$. Hence, the algorithm adds $(v, w)$ to $\MS(H, \pi)$. Note that all \start{} elements incident to $v$ after $(v, w)$ in the permutation cannot be in $\MS(H, \pi)$ since \start{} elements create a maximal matching. By definition of $R$, the first two \extend{} elements incident to $v$ are between $v$ and $U_v$. Let $(v, u_1)$ and $(v, u_2)$ be these two edges ($u_1 \neq u_2$ since there is one \extend{} copy). Therefore, one of $(v, u_1)$ and $(v, u_2)$ must be added to $\MS(H, \pi)$ and no other \extend{} element incident to $v$ can be added to $\MS(H, \pi)$ since \extend{} elements create a maximal matching.
\end{proof}

\begin{observation}\label{obs: large-matching-R}
It holds that $$\left(\frac{1}{2} + \delta \right)\mu(H[V_1 \cup R]) \leq \E_\pi|\MS(H, \pi) \cap \left((V_1 \cup R) \times (V_1 \cup R)\right)|  \leq \mu(H[V_1 \cup R]).$$
\end{observation}
\begin{proof}
Note that by \Cref{obs: good-vertices}, all vertices in $V_2 \setminus R$ have incident \start{} and \extend{} elements in $\MS(H,\pi)$ which appear before edges between $V1 \cup R$ and $V_2 \setminus R$ in the permutation. Hence, none of these edges can be added to the subgraph $\MS(H, \pi)$ in \Cref{alg: algorithm}. Since the permutation over edges in $H[V_1 \cup R]$ is uniformly at random, using \Cref{thm:apx}, we have the given bound
\end{proof}

Next, we provide an upper bound for the size of $R$.

\begin{observation}\label{obs: R-size}
It holds that $\E_\pi|R| \leq  \frac{n}{2\log n}$.
\end{observation}
\begin{proof}
For vertex $v \in V_2$ and a random permutation $\pi$ over $E_H$, the first incident \start{} element is between $v$ and $U_v$, with probability of at least $\frac{K\gamma}{(n + \gamma)K} \geq 1 - \frac{1}{4\log n}$. Furthremore, with probability $\frac{\gamma(\gamma - 1)}{(n+\gamma)(n+\gamma - 1)} \geq 1 - \frac{1}{4 \log n}$, the first two \extend{} elements are between $v$ and $U_v$. Since these events are independent, the probability of $v$ not being a bad vertex is at least
$$
\left(1 - \frac{1}{4\log n}\right)^2 \geq 1 - \frac{1}{2\log n}
.$$

Therefore, $|R|$ is at most $\frac{n}{2\log n}$ in expectation over a random permutation.
\end{proof}

With the intuition that a few vertices in $V_2$ are matched to vertices in $V_1 \cup V_2$, our goal is to estimate the number of vertices that have a matching edge in $V_1$. Since we have an upper bound on the size of the $R$,  we are able to estimate the number of matching edges in $H[V_1] = H[G]$. In order to count the number of matching edges in $V_1$, we need to run the vertex oracle for vertices of $V_1 \cup V_2$. We show that the expected running time of vertex oracle on a random vertex of $V_1 \cup V_2$ is $\Tilde{O}(n^{1+\epsilon})$.

\begin{claim}\label{clm: v1v2query}
Let $v$ be a random vertex in $V_1 \cup V_2$ and $\pi$ be a random permutation over $E_H$. Then
$$
    \E_{v\sim (V_1 \cup V_2),\pi}[F(v, \pi)] = \tilde{O}(n^{1+\epsilon}).
$$
\end{claim}
\begin{proof}
By \Cref{thm: query-complexity}, we have 
\begin{align*}
    \E_{v\sim V_H,\pi}[F(v, \pi)] = O\left(K\frac{|E_H|}{|V_H|}\log^4 |V_H|\right).
\end{align*}

Summing over all vertices of $V_H$, we get
\begin{align*}
    \sum_{v \in V_H}\E_{\pi}[F(v, \pi)] = O(K|E_H| \cdot \log^4 |V_H|) = \tilde{O}(n^{2+\epsilon}),
\end{align*}
because $|V_H|=O(n^2)$, $|E_H| = O(n^2 + n\gamma)$, $K = \tilde{O}(n^\epsilon$), and $\gamma =(4\log n)\cdot n$. Therefore,
$$
    \E_{v\sim (V_1 \cup V_2),\pi}[F(v, \pi)] \leq \left( \sum_{v \in V_H}\E_{\pi}[F(v, \pi)] \right) / |(V_1 \cup V_2)| = \tilde{O}(n^{1+\epsilon}).\qedhere
$$
\end{proof}
\begin{claim}\label{clm: v1query}
Let $v$ be a random vertex in $V_1$ and $\pi$ be a random permutation over $E_H$. Then
$$
    \E_{v\sim V_1, \pi}[F(v, \pi)] = \tilde{O}(n^{1+\epsilon}).
$$
\end{claim}
\begin{proof}
Proof follows by combining \Cref{clm: v1v2query} and $|V_1| = |V_2|$.
\end{proof}

\begin{figure}[h]
\begin{algenv}{Final algorithm for adjacency matrix.}{alg: adjacency-matrix}
    Let $H = (V_H, E_H)$ as described above.

    $r \gets 48 \cdot \log^3 n$.
    
    Sample $r$ vertices $u_1, u_2, \ldots, u_r$ uniformly at random from $V_1$ with replacement.
    
    Run vertex oracle for each $u_i$ and let $S_i$ and $E_i$ be the indicator that $u_i$ has an incident \start{} and \extend{} elements that appear in output of \cref{alg: algorithm}.
    
    \For{$i$ in $1 \ldots r$}{ \label{line: forloop2}
    
    $X_i \gets 0$
    
    \If{$S_i = 1$}{
        
        Let $w_i$ be the other endpoint of \start{} edge incident to $u$.
        
        \lIf{$w_i \in V_1$}{$X_i \gets 1$}
    }    
    \If{$S_i = 0$ and $E_i = 1$}{
        Let $w_i$ be the other endpoint of \extend{} element incident to $u$.
        
        Run vertex oracle for $w_i$ (with same $\pi$) and let $S'_i$ be the indicator that $w_i$ has an incident \start{} edge.
        
        \lIf{$S'_i = 0$}{\Continue}
        
        Let $x_i$ be other endpoint of \start{} element incident to $w_i$.
        
        Run vertex oracle for $x_i$ (with same $\pi$) and let $E'_i$ be the indicator that $x_i$ has an incident \extend{} element.
        
        \lIf{$E'_i = 0$}{\Continue}
        
        Let $y_i$ be other endpoint of \extend{} element incident to $x_i$.
        
        Run vertex oracle for $y_i$ (with same $\pi$) and let $S''_i$ be the indicator that $y_i$ has an incident \start{} element.
        
        \lIf{$S''_i = 0$ and $w_i \in V_1$}{$X_i \gets 1$}
    }
    }

Let $X = \sum_{i \in [r]}X_i$ and $f = X/r$.

Let $\tilde{\mu} = fn/2 - \frac{n}{4 \log n}$.
    
\Return $\tilde{\mu}$

\end{algenv}
\end{figure}

\begin{lemma}\label{lem: multi-add-error-bound-matrix}
Let $\tilde{\mu}$ be the output of \cref{alg: adjacency-matrix}. With high probability, 
$$ \left(\frac{1}{2} + \delta \right)\mu(G) - \frac{n}{\log n} \leq \tilde{\mu} \leq \mu(G).$$
\end{lemma}
\begin{proof}
Let $\hat{M}(H, \pi)$ be the intersection of edges between $V_1$ and set of edges of the matching that is created by augmenting the length-three augmenting paths of $M \cup S$. We claim 
\begin{align}\label{eq: mhat-range}
    \left(\frac{1}{2} + \delta \right)\mu(H[V_1]) - \frac{n}{2\log n} \leq \E_{\pi}|\hat{M}(H, \pi)|  \leq \mu(H[V_1]).
\end{align}
By combining \Cref{obs: large-matching-R} and \Cref{obs: R-size}, imply that there are at most $\frac{n}{2\log n}$ edges in output of algorithm in $H[V_1 \cup R]$ with at least one endpoint in $R$. Therefore, by combining \Cref{rem: three-augmenting}, \Cref{thm:apx}, and the bound for number of edges with at least one endpoint in $R$, we have the first inequality. Furthermore, since $\hat{M}(H, \pi)$ is a matching of $H[V_1]$, we have the second inequality.

By definition, $X_i$ is the indicator of the event that a vertex $i \in V_1$ is matched in the output of \Cref{alg: algorithm} to another vertex in $V_1$. Since the number of matched vertices is twice the number of matching edges,  
\begin{align*}
    \E[X_i] = \Pr[X_i = 1] = \frac{2 \E|\hat{M}(H, \pi)|}{n}.
\end{align*}
Similarly,  
\begin{align}\label{eq: X-expected2}
    \E[X] = \frac{2r \E|\hat{M}(H, \pi)|}{n}.
\end{align}

Since $X$ is the sum of $r$ independent Bernoulli random variables, the Chernoff bound (\Cref{prop:chernoff}) implies
\begin{align*}
    \Pr[|X - \E[X]| \leq \sqrt{6 \E[X] \log n}] \leq 2 \exp \left(-\frac{6 \E[X] \log n}{3 \E[X]} \right) = \frac{2}{n^2}.
\end{align*}

Note that $fn = Xn / r$, so using the above equation, $fn$ is in the following range with probability $1 - 2/n^2$:
\begin{align*}
    fn \in \frac{n(\E[X] \pm \sqrt{6 \E[X] \log n)}}{r} & = \frac{n\E[X]}{r} \pm \sqrt{\frac{6 n^2\E[X] \log n}{r^2}}\\
    & = 2\E|\hat{M}(H, \pi)| \pm \sqrt{\frac{12n\E|\hat{M}(H, \pi)| \log n}{r}} & (\text{By }(\ref{eq: X-expected2}))\\
    & = 2\E|\hat{M}(H, \pi)| \pm \sqrt{\frac{n\E|\hat{M}(H, \pi)|}{4\log^2 n}} & (\text{Since } r = 48 \cdot \log^3 n) \\
    & \in 2\E|\hat{M}(H, \pi)| \pm \frac{n}{2 \log n} & (\text{Since } \E|\hat{M}(H, \pi)| \leq n).
\end{align*}
Now, $\tilde{\mu} = fn/2 - \frac{n}{4 \log n}$. Therefore,
\begin{align*}
    \E|\hat{M}(H, \pi)| - \frac{ n}{2 \log n} \leq \tilde{\mu} \leq \E|\hat{M}(H, \pi)|.
\end{align*}
Combining the above range with (\ref{eq: mhat-range}) implies
\begin{align*}
    \left(\frac{1}{2} + \delta \right)\mu(G) - \frac{ n}{\log n} \leq \tilde{\mu} \leq \mu(G). \qquad\qedhere
\end{align*}
\end{proof}

\begin{lemma}\label{lem: matrix-time}
\Cref{alg: adjacency-matrix} runs in $\Tilde{O}(n^{1+\epsilon})$ with high probability.
\end{lemma}
\begin{proof}
First, we show that the iteration in Line \ref{line: forloop2}, for each of $r$ sampled vertices from $V_1$ takes $\Tilde{O}(n^{1+\epsilon})$ time in expectation. With the same argument as \Cref{lem: adjcacency-list-time}, the vertex oracle for a vertex $u$ can be called a constant number of times. Moreover, by \Cref{clm: v1v2query} and \Cref{clm: v1query}, the vertex oracle takes $\Tilde{O}(n^{1+\epsilon})$ time for a random vertex in $V_1 \cup V_2$ and a random permutation (since set $R$ is a uniformly at random set in $V_2$). Therefore, if we run the vertex oracle constant time for a fixed vertex, still the expected running time will be $\Tilde{O}(n^{1+\epsilon})$.

In order to achieve a high probability bound on the running time, similar to the proof of \Cref{lem: adjcacency-list-time}, we run $\Theta(\log n)$ instances of the algorithm simultaneously. Using the same argument, the first instance terminates with probability $1 - 1/\text{poly}(n)$ in $\Tilde{O}(n^{1+\epsilon})$.
\end{proof}

\begin{proof}[Proof of \Cref{thm: main-theorem}]\leavevmode
\begin{itemize}
    \item By combining \Cref{lem: multi-error-bound-adj} and \Cref{lem: adjcacency-list-time}, we get multiplicative factor of $(\frac{1}{2} + \delta)$ in the adjacency list model in $\widetilde{O}(n+\Delta^{1+\epsilon})$ time.
    \item By combining \Cref{lem: multi-add-error-bound-adj} and \Cref{lem: adjcacency-list-additive-time}, we get multiplicative-additive factor of $(\frac{1}{2} + \delta, o(n))$ in the adjacency list model in $\widetilde{O}(\bar{d} \cdot \Delta^\epsilon)$ time.
    \item By combining \Cref{lem: multi-add-error-bound-matrix} and \Cref{lem: matrix-time}, we get multiplicative-additive factor of $(\frac{1}{2} + \delta, o(n))$ in the adjacency matrix model in $\Tilde{O}(n^{1+\epsilon})$ time. For a constant $\epsilon$, running this algorithm for a slightly smaller value of $\epsilon$ allows us to get rid of polylogarithmic factor in the running time and achieve an $O(n^{1+\epsilon})$ time algorithm. \qedhere
\end{itemize}
\end{proof}

%% file: conclusion.tex
\section{Conclusion}

We presented a new algorithm for finding a  $(\frac{1}{2}+\Omega(1))$-approximate maximum matching and showed how its output size can be estimated in sublinear time. The algorithm gives a multiplicative $(\frac{1}{2}+\Omega(1))$-approximation of the size of maximum matching in the adjacency list model that runs in $\widetilde{O}(n+\Delta^{1+\epsilon}) = O(n^{1+\epsilon})$ time, where constant $\epsilon > 0$ can be made arbitrarily small. This is  the first algorithm that beats half-approximation in $o(n^2)$ time for all graphs. 

Given that the barrier of 1/2 is now broken, a natural question is what is the best approximation achievable in $n^{2-\Omega(1)}$ time. Is it possible to get a, say, .51-approximation? On the flip side, a lower bound ruling out say a $.9999$-approximation in $n^{2-\Omega(1)}$ time would also be extremely interesting.

%% file: appendix-proofs.tex
\section{Deferred Proofs}\label{apx:missing-proofs}
\begin{proof}[Proof of \cref{cl:many-length-three}]
     For any odd $i$, let $t_i$ be the number of augmenting paths of length $i$ for $M$ in $M \oplus M^\star$. Note that since $M$ is maximal, we have $t_1 = 0$. Therefore,
     $$
        |M^\star| - |M| = \sum_{i \geq 1} t_{2i+1} \leq t_3 + \sum_{i \geq 2} \frac{1}{2} i t_{2i+1} \leq \frac{1}{2} t_3 + \frac{1}{2} \sum_{i \geq 1} i t_{2i+1} = \frac{1}{2} t_3 + \frac{1}{2}|M|.
     $$
     Moving the terms and using the assumption of $|M| < (\frac{1}{2}+\delta)|M^\star|$, we get
     $$
        |M| - t_3 \leq |M| - (2|M^\star| - 3 |M|) = 4 |M| - 2\mu(G) \leq (2+4\delta)\mu(G) - 2\mu(G) = 4\delta \mu(G).\qedhere
     $$
\end{proof}

\section{Implementation Details}\label{sec: implementation}
In this section, we provide an implementation of our vertex and edge oracle. The idea is similar to the \cite[Appendix~A]{behnezhad2021}. The main difference is that instead of having at most one edge between two vertices, here we have $K+1$ edges where $K$ of them correspond to \start{} copies and one of them corresponds to \extend{} copy. The idea is to generate a random permutation $\pi$ locally and sort edges based on $\pi$ to create the permutation.


Note that in the vertex oracle and edge oracle, when an \start{} incident element appears in $\MS(G, \pi)$, the algorithm no longer queries on the \start{} incident elements (the same also holds for \extend{} elements). Therefore, instead of having one graph, we assume that we have two graphs $G_s$ and $G_e$ that are isomorphic to $G$. Also, we assume that each edge in $G_s$ has $K$ copies that is corresponds to the number of \start{} elements. In other words, graph $G_s$ is the a graph made by all \start{} elements of $G$ and $G_e$ is the a graph made by all \extend{} elements of $G$.

Let $\lowest_{G_t}(u, i)$ for $t \in \{s, e\}$ be a procedure that returns a pair of an edge $e$ in $G_t$ and its ranking such that $e$ is the $i$-th lowest rank edge incident to $u$ in $G_t$. We use the same implementation of $\lowest_{G_t}$ procedure as the \cite[Appendix~A]{behnezhad2021} (note that in this paper, the procedure only returns the edge. However, in the implementation of the \lowest, they compute the edge ranking and we can simply return the edge ranking). Let $\deg_{G_t}(u)$ be the degree of vertex $u$ in graph $G_t$. By definition of the $G_s$ and $G_e$, we have that $\deg_{G_s}(u) = K\deg_G(u)$ and $\deg_{G_e}(u) = \deg_G(u)$.

\begin{claim}[\cite{behnezhad2021}]\label{clm: lowest-procedure}
Let $u$ be a vertex and suppose that we call procedure $\lowest_{G_t}(u, i)$ for all $1 \leq i \leq j$. The total time to implement all these calls is $\Tilde{O}(j)$ with high probability for all $u \in V$.
\end{claim}

We present the implementation of our oracles in the following three algorithms. Note that we can generate vertex colors (i.e. $c_v$ for vertex $v$) on the fly when it is needed. Also, we can toss a coin with a probability $1-p$ for each edge to determine if it is a frozen edge or not. Therefore, once we have the rank of edges in our oracles, we can distinguish whether an edge is frozen or not based on the \Cref{def: freeze}. However, to make algorithms easier to read, we do not include the technical details of this part.

\begin{algenv}{Implementation of the vertex oracle $\VO(u)$.}{alg: vertex_oracle_implementation}
$j_s \gets 1, j_e \gets 1$

$v_s, \pi_s \gets \lowest_{G_s}(u, j_s)$

$v_e, \pi_e \gets \lowest_{G_e}(u, j_e)$

$ST \gets \false$, $EX \gets \false$

\While{$j_s \leq \deg_{G_s}(u)$ or $j_e \leq \deg_{G_e}(u)$}{
    \If{$\pi_s < \pi_e$}{
        $X \gets \start{}$
        
        $\ell = ((u, v_s), X)$
    }
    \Else{
        $X \gets \extend{}$
    
        $\ell = ((u, v_e), X)$
    }
    
    \lIf{$X = \extend{}$ and $c_u = c_{v_e}$}{\Continue}
    
    \If{$ST = \false$ and $X = \start$ and $\EOS(\ell, v_s) = \true$}{
            $ST \gets \true$
            
            $j_s \gets j_s + 1$
            
            \If{$j_s \leq \deg_{G_s}(u)$}{
            $v_s, \pi_s \gets \lowest_{G_s}(u, j_s)$
            }
            \Else{
                $v_s, \pi_s \gets \infty, \infty$
            }
            
    }
    \If{$EX = \false$ and $X = \extend$ and $\EOE(\ell, v_e, ST) = \true$}{
            $EX \gets \true$
            
            $j_e \gets j_e + 1$
            
            \If{$j_e \leq \deg_{G_e}(u)$}{
            $v_e, \pi_e \gets \lowest_{G_e}(u, j_e)$
            }
            \Else{
                $v_e, \pi_e \gets \infty, \infty$
            }
    }
    \Return $ST, EX$
}
    
\end{algenv}

\begin{algenv}{Implementation of the edge oracle for \start{} elements $\EOS(\ell ,u)$.}{alg: start_oracle_implementation}
    \lIf{\textup{EOS}$(\ell, u)$ is already computed}{\Return the computed answer} \label{line: caching}
    
    $j \gets 1$
    
    $w, \pi_w \gets \lowest_{G_s}(u, j)$
    
    \While{$w \neq v$}{
        $\ell' \gets ((u, w), \start{})$
    
        \lIf{$\EOS(\ell', w) = \true$}{\Return \false}
        
        $j \gets j + 1$
        
        $w \gets \lowest_{G_s}(u, j)$
    }
    
    \Return \true

\end{algenv}

\begin{algenv}{Implementation of the edge oracle for \extend{} elements $\EOE(\ell, u, ST_w)$.}{alg: extend_oracle_implementation}

\lIf{\textup{EOE}$(\ell, u, ST_w)$ is already computed}{\Return the computed answer} \label{line: caching}

    $j_s \gets 1, j_e \gets 1$

    $v_s, \pi_s \gets \lowest_{G_s}(u, j_s)$

    $v_e, \pi_e \gets \lowest_{G_e}(u, j_e)$

    $ST_u \gets \false$

    \If{$\pi_s < \pi_e$}{
        $w \gets v_s, X' \gets \start{}$
    }
    \Else{
        $w \gets v_e, X' \gets \extend{}$
    }
        
    $\ell' \gets ((u, w), X')$

    \While{$w \neq v$ or $X' \neq X$}{
        \lIf{$X' = \extend{}$ and $c_u = c_{w}$}{\Continue}

        \If{$X' = \start$}{
            \If{$ST_u = \false$, and $\EOS(\ell', w) = \true$}{
                $ST_u \gets \true$
            
                \lIf{$\ell'$ is frozen}{\Return \false}
            
                \lIf{$ST_w = \true$}{\Return \false} 
            }
            
            $j_s \gets j_s + 1$
            
            \If{$j_s \leq \deg_{G_s}(u)$}{
            $v_s, \pi_s \gets \lowest_{G_s}(u, j_s)$
            }
            \Else{
                $v_s, \pi_s \gets \infty, \infty$
            }
            
        }

        \If{$X' = \extend{}$}{
        \lIf{$\EOE(\ell', w, ST_u) = \true$}{\Return \false}
            
            $j_e \gets j_e + 1$
            
            \If{$j_e \leq \deg_{G_e}(u)$}{
            $v_e, \pi_e \gets \lowest_{G_e}(u, j_e)$
            }
            \Else{
                $v_e, \pi_e \gets \infty, \infty$
            }
        }
    
        \If{$\pi_s < \pi_e$}{
            $w \gets v_s, X' \gets \start{}$
        }
        \Else{
            $w \gets v_e, X' \gets \extend{}$
        }
        
        $\ell' \gets ((u, w), X')$
        
    }
    
    \Return \true

\end{algenv}

\smallskip\smallskip

\oracleImplementation

\begin{proof}
    Note that in \Cref{alg: vertex_oracle_implementation}, we only call \lowest{} procedure for the neighbors that we recursively call edge oracle for them. Similarly, in \Cref{alg: start_oracle_implementation} and \Cref{alg: extend_oracle_implementation} we only call \lowest{} procedure for incident elements that we will recursively call the edge oracle on them. Therefore, by \Cref{clm: lowest-procedure}, the total time spent on \lowest{} procedure calls is $\Tilde{O}(F(v, \pi))$ for a random permutation $\pi$ and every vertex $v$ with high probability.
\end{proof}